\documentclass[10pt]{article}

\usepackage{url}
\usepackage{mathtools}
\usepackage{amssymb}
\usepackage{amsthm}
\usepackage{amsmath}
\usepackage{empheq}
\usepackage{latexsym}
\usepackage{enumitem}
\usepackage{eurosym}
\usepackage{dsfont}
\usepackage{appendix}
\usepackage{color} 
\usepackage{comment}
\usepackage[unicode]{hyperref}
\usepackage[utf8]{inputenc}
\usepackage[T1]{fontenc}
\usepackage{geometry}
\usepackage{multirow}
\usepackage{todonotes}
\usepackage{lmodern}
\usepackage{anyfontsize}
\usepackage{stmaryrd}
\usepackage{cleveref}
\usepackage{natbib}
\usepackage[english]{babel}
\usepackage[english=british]{csquotes}
\usepackage{tikz}
\usepackage{pgfplots}
\usepackage{graphicx}
\usepackage{float}
\usepackage{caption}
\usepackage{subcaption}
\pgfplotsset{compat=1.17} 

\setcitestyle{numbers,open={[},close={]}}

\definecolor{red}{rgb}{0.7,0.15,0.15}
\definecolor{green}{rgb}{0,0.5,0}
\definecolor{blue}{rgb}{0,0,0.7}

\hypersetup{colorlinks, linkcolor={red},citecolor={green}, urlcolor={blue}}
			
\makeatletter \@addtoreset{equation}{section}

\newtheorem{theorem}{Theorem}[section]
\newtheorem{assumption}[theorem]{Assumption}

\newtheorem{lemma}[theorem]{Lemma}

\newtheorem{definition}[theorem]{Definition}
\newtheorem{remark}[theorem]{Remark}

\DeclareUnicodeCharacter{014D}{\=o}
\def \A{\mathbb{A}}

\def \E{\mathbb{E}}
\def \F{\mathbb{F}}
\def \G{\mathbb{G}}
\def \H{\mathbb{H}}

\def \M{\mathbb{M}}
\def \N{\mathbb{N}}

\def \P{\mathbb{P}}

\def \R{\mathbb{R}}
\def \S{\mathbb{S}}

\def \V{\mathbb{V}}

\def\Ac{{\cal A}}
\def\Bc{{\cal B}}
\def\Cc{{\cal C}}
\def\Dc{{\cal D}}
\def\Ec{{\cal E}}
\def\Fc{{\cal F}}
\def\Gc{{\cal G}}
\def\Hc{{\cal H}}

\def\Nc{{\cal N}}
\def\Oc{{\cal O}}
\def\Pc{{\cal P}}
\def\Qc{{\cal Q}}

\def\Sc{{\cal S}}

\def\Vc{{\cal V}}

\def\eps{\varepsilon}

\def\drm{\mathrm{d}}

\title{Incentives, lockdown, and testing: from Thucydides' analysis to the COVID-19 pandemic\footnote{The authors acknowledge the supports of the ANR projects PACMAN ANR-16-CE05-0027 and ReLISCoP ANR-21-CE40-0001.} }

\author{Emma {\sc Hubert} \footnote{Operation Research and Financial Engineering, Princeton University, USA,  eh3988@princeton.edu} \and Thibaut {\sc Mastrolia} \footnote{Industrial Engineering and Operations Research, UC Berkeley, USA, mastrolia@berkeley.edu.}\and Dylan {\sc Possama\"{i}} \footnote{ETH Z\"urich, Mathematics department, Switzerland, dylan.possamai@math.ethz.ch}\and Xavier {\sc Warin}\footnote{EDF R\&D and FiME, Laboratoire de Finance des March\'es de l’\'Energie \url{(www.fime-lab.org}), xavier.warin@edf.fr}}
 
\date{\today}

\begin{document}

\maketitle

\begin{abstract} 
In this work, we provide a general mathematical formalism to study the optimal control of an epidemic, such as the COVID-19 pandemic, \textit{via} incentives to lockdown and testing. In particular, we model the interplay between the government and the population as a principal--agent problem with moral hazard, \textit{à la} \citeauthor*{cvitanic2018dynamic} \cite{cvitanic2018dynamic}, while an epidemic is spreading according to dynamics given by compartmental stochastic SIS or SIR models, as proposed respectively by \citeauthor*{gray2011stochastic} \cite{gray2011stochastic} and \citeauthor*{tornatore2005stability} \cite{tornatore2005stability}. More precisely, to limit the spread of a virus, the population can decrease the transmission rate of the disease by reducing interactions between individuals. However, this effort---which cannot be perfectly monitored by the government---comes at social and monetary cost for the population. To mitigate this cost, and thus encourage the lockdown of the population, the government can put in place an incentive policy, in the form of a tax or subsidy. In addition, the government may also implement a testing policy in order to know more precisely the spread of the epidemic within the country, and to isolate infected individuals. In terms of technical results, we demonstrate the optimal form of the tax, indexed on the proportion of infected individuals, as well as the optimal effort of the population, namely the transmission rate chosen in response to this tax. The government's optimisation problems then boils down to solving an Hamilton--Jacobi--Bellman equation. Numerical results confirm that if a tax policy is implemented, the population is encouraged to significantly reduce its interactions. If the government also adjusts its testing policy, less effort is required on the population side, individuals can interact almost as usual, and the epidemic is largely contained by the targeted isolation of positively-tested individuals.	

\medskip


\medskip

\end{abstract}

\section{Introduction}

Starting around 430 BC, and known as the first historically epidemic, the plague of Athens killed between a quarter and a third of Athenians, as reported by Thucydides.
He analysed the consequences of this epidemic, and concluded that it had led a moral upheaval for the Athenians, faced with the complete lack of any useful cure. In the end, the disease was only stopped thanks to the development of a natural immunity within the population, during the first four years of the epidemic phase. Concerning the spread of the disease itself, Thucydides wrote the following
\blockquote[{\citeauthor*{jowett1900thucydides} \cite[Volume I, Book II, pp. 138]{jowett1900thucydides}}][.]{When they were afraid to visit one another, the sufferers died in their solitude, so that many houses were empty because there had been no one left to take care of the sick; or if they ventured they perished, especially those who aspired to heroism. For they went to see their friends without thought of themselves and were ashamed to leave them, at a time when the very relations of the dying were at last growing weary and ceased even to make lamentations, overwhelmed by the vastness of the calamity}

From this analysis, we can already isolate three fundamental questions that need to be addressed whenever an unknown epidemic occurs.
\begin{enumerate}
\item[$(1)$] How can one model a disease with only parsimonious information on how it is spreading among the population?

\vspace{-0.6em}
\item[$(2)$] How can one solve the Gordian knot associated to interactions within the population: enjoying on the one hand the presence of others and avoiding solitude, and on the other hand dramatically spreading the disease?

\vspace{-0.6em}
\item[$(3)$] How can governments and decision-makers incentivise people in order to better control the spread of the epidemic?
\end{enumerate}

\textbf{Choosing a relevant epidemic model.}
The first question is naturally linked to several strands of fundamental research, both for mathematicians and physicians, dealing with the problem of choosing a relevant epidemic model. The paternity of the first mathematical model designed to describe the evolution of an epidemic seems to be attributed to \citeauthor*{bernoulli1760essai}, who proposed one for smallpox as early as 1760 in \cite{bernoulli1760essai}. However, other early mathematical approaches were used to study various types of epidemics and their consequences, for example by \citeauthor*{farr1840causes} in 1840, who applied mathematics to death records during a smallpox epidemic in England in \cite{farr1840causes}, and whose work can be considered as a starting point of the field.
Nevertheless, the real mathematical development of the theory had to wait for the 20th century, with fundamental contributions for the development of deterministic models by \citeauthor*{hamer1906epidemic} \cite{hamer1906epidemic}, \citeauthor*{ross1910prevention} \cite{ross1910prevention}, and later \citeauthor*{bartlett1949some} \cite{bartlett1949some} who proposed one of the first general investigations of the evolution of deterministic interacting systems, which was then applied to epidemiology in \citeauthor*{kendall1956deterministic} \cite{kendall1956deterministic}. It was rapidly noticed that deterministic models were insufficient to account for the uncertainty associated with the disease spreading, and the technical difficulties usually encountered for its detection. This acknowledgement helped nurturing the development of stochastic models, whose first instance seems to be traced back to \citeauthor*{m1925applications} \cite{m1925applications}. For a precise comparison between deterministic and stochastic models in discrete-time settings as well as more historical details, we refer our readers to \citeauthor*{bailey1975mathematical} \cite{bailey1975mathematical}, and to \citeauthor*{allen2008introduction} \cite{allen2008introduction} for more up-to-date references and an overview of recent epidemiological models.

\medskip

We will now describe some models, belonging to the general class of compartmental models, and which will be at the heart of our work. The first one considers a sort of worst-case scenario, in which an immunity is not developed after infection. Such models have been coined SIS (for Susceptible--Infected--Susceptible), and consider a population divided into two groups: susceptible individuals interact with infected ones, and therefore move from one class to the other repeatedly. This model was first discussed in \citeauthor*{weiss1971asymptotic} \cite{weiss1971asymptotic}, and then extended by \citeauthor*{naasell1996quasi} \cite{naasell1996quasi}, who found the quasi-stationary distribution of a continuous-time stochastic SIS model with no births nor deaths. In this work, the stochastic SIS model we will focus on is defined as a solution to a bi-dimensional SDE driven by a single Brownian motion, as proposed by \citeauthor*{gray2011stochastic} \cite{gray2011stochastic}. Alternatively to this quite pessimistic scenario, one can assume that an immunity will appear after infection, thus adding a third state: the recovered individuals, who have been cured and developed antibodies. Introduced originally by \citeauthor*{kermack1927contribution} \cite{kermack1927contribution}, this so-called SIR model was studied in depth by \citeauthor*{anderson1979population} \cite{anderson1979population} in a deterministic setting, while stochastic perturbations were introduced by \citeauthor*{beretta1998stability} \cite{beretta1998stability}. To be consistent with our choice for the SIS model, we will consider the stochastic SIR model proposed by \citeauthor*{tornatore2005stability} \cite{tornatore2005stability}. It should be noted that there is a wide variety of formulations of stochastic SIS and SIR models, which makes it impossible to list them all here. We will simply mention the works by \citeauthor*{britton2019stochastic} \cite{britton2019stochastic}, \citeauthor*{dieu2016classification} \cite{dieu2016classification}, \citeauthor*{du2020permanence} \cite{du2020permanence}, \citeauthor*{jiang2011asymptotic} \cite{jiang2011asymptotic} and \citeauthor*{schreiber2021extinction} \cite{schreiber2021extinction}, on the study of the long-term behaviour of this type of stochastic models, thus answering the question whether or not the epidemic can be controlled.

\medskip
\textbf{On the control of an epidemic.}
In the aforementioned classical compartmental models, the infection grows into the population through an incidence rate $\beta$, and proportionally to the product of the number of susceptible and infected individuals, as already discussed in the work by \citeauthor*{wilson1945law} \cite{wilson1945law}, or in the \textit{Reed--Frost theory}, revisited for instance by \citeauthor*{abbey1952examination} \cite{abbey1952examination}. In the absence of a cure or a vaccine, this transmission rate appears as the only control variable for individuals or public institutions to reduce the spread of an epidemic. Our take on the second main question will therefore be from a control-theoretic perspective. At the heart of this approach is the simple idea that when faced with an epidemic, a perfectly rational population will try to find an equilibrium interaction rate, balancing the need to still connect with others, and the natural fear of spreading the infection itself. This is by no means a new point of view, and papers discussing the use of formal control theory in epidemiology can be dated back to the 70s, see among others \citeauthor*{taylor1968some} \cite{taylor1968some}, \citeauthor*{abakuks1973optimal} \cite{abakuks1973optimal}, \citeauthor*{morton1974optimal} \cite{morton1974optimal}, \citeauthor*{wickwire1975optimal} \cite{wickwire1975optimal}, or \citeauthor*{sethi1978optimal} \cite{sethi1978optimal}. More recently and closer to our purpose, we can refer to \citeauthor*{behncke2000optimal} \cite{behncke2000optimal}, \citeauthor{riley2003transmission} \cite{riley2003transmission}, who studied the impact of the control of transmission rate on the 2002--2004 SARS outbreak in Hong Kong and on the ways to interfere with the disease's spread, \citeauthor*{hansen2011optimal} \cite{hansen2011optimal}, and more broadly to the monograph by \citeauthor*{lenhart2007optimal} \cite{lenhart2007optimal}. 

\medskip
As should be expected, a significant part of the recent literature on the COVID-19 pandemic has also adopted this control point of view, and such lockdown measures as well as their medical, societal, and economical impacts are discussed by, among others, \citeauthor*{anderson2020will} \cite{anderson2020will}, \citeauthor*{bayraktar2021macroeconomic} \cite{bayraktar2021macroeconomic}, \citeauthor*{charpentier2020covid} \cite{charpentier2020covid}, 
\citeauthor{ferguson2020report} \cite{ferguson2020report},
\citeauthor*{fowler2020effect} \cite{fowler2020effect}, \citeauthor*{grigorieva2020optimal} \cite{grigorieva2020optimal}, \citeauthor*{hatchimonji2020trauma} \cite{hatchimonji2020trauma}, \citeauthor*{kantner2020beyond} \cite{kantner2020beyond}, \citeauthor*{piguillem2020optimal} \cite{piguillem2020optimal}, or \citeauthor*{wilder2020can} \cite{wilder2020can}. The previous papers take the point of view of a government acting as a central planner, in the sense that it can impose on the population to control the epidemic in a way which is beneficial to the population as a whole. However, though it seems reasonable to assume that at least some individuals, by being afraid of getting sick, will naturally decrease their interaction rates, it would clearly be a stretch to consider that all individuals will follow the governmental's recommendations. This individuals' point of view have been considered by \citeauthor{reluga2013equilibria} in \cite{reluga2010game} and \cite{reluga2013equilibria}, as well as by \citeauthor*{li2017provisioning} \cite{li2017provisioning}, thus introducing game theory in epidemiologic models, or more recently by \citeauthor*{elie2020contact} \cite{elie2020contact} for the case of COVID-19.

\medskip
\textbf{Introducing the notion of incentives.} In light of the issues we have raised, a natural conclusion was, at least for us, that even if a control-theoretic approach to mitigate the impact of an epidemic is clearly desirable, there is {\it a priori} no evidence that in face of clear public policies, a population will directly adopt a social distancing behaviour leading to an optimal transmission rate for the welfare of the society. Moreover, in the absence of a system allowing to actually keep track of the level of interaction within the population, governments are faced with a clear situation of moral hazard: it is impossible for large countries to ensure the application of such isolation measures, and therefore it is unfeasible to have an absolute control on the behaviour of all individuals and their interactions.\footnote{Several countries worldwide have decided to use contact-tracing tools, such as mobile phone apps, designed to help tracking down subsequent exposures after an infected individual is identified, see for instance \citeauthor*{cho2020contact} \cite{cho2020contact}, or \citeauthor*{reichert2020privacy} \cite{reichert2020privacy}. Using these would in principle erase any possibility or moral hazard, provided that all the population uses the app, and that testing is organised on a massive scale. Even admitting that this would be the case, it remains that these tools have raised complex issues of privacy, see \citeauthor*{ienca2020responsible} \cite{ienca2020responsible} or \citeauthor*{park2020information} \cite{park2020information}. In any case, the incentive-based approach we propose can always be considered as a useful complement to any other adopted strategy.} Consequently, an incentive policy should also be calibrated by governments in order to get a better control on the spread of the disease. This, as expected, leads us to our third question, which is where our approach departs significantly from the extant literature. 
Indeed, to our knowledge, the literature on optimal incentives to counter moral hazard in the context of an epidemic is very sparse. Some authors, for instance \citeauthor*{valeeva2007incentive} \cite{valeeva2007incentive} or \citeauthor*{gramig2005model} \cite{gramig2005model,gramig2009livestock}, study disease spreading through the lens of asymmetry of information, but they are mostly interested in livestock related diseases, where producers have private information on preventive measures they may have adopted, prior to contamination ({\it ex ante} moral hazard), and may or may not declare whether their herd is infected after contamination ({\it ex post} adverse selection). A paper by \citeauthor*{francis2004optimal} \cite{francis2004optimal} discusses the optimal taxes/subsidies to encourage vaccination during the flu season. More closely related to our principal--agent formulation, \citeauthor*{carmona2021finite} consider in \cite[Section 5]{carmona2021finite} an application to the containment of an epidemic of their moral hazard theory for agents interacting through a finite state mean-field game. Finally, an approach similar to ours, but which takes into account mean-field type interactions between individuals within the population, has been developed concurrently and independently of the present paper by \citeauthor*{aurell2020optimal} \cite{aurell2020optimal}.

\medskip
\textbf{Principal--agent approach and technical results.} We thus propose to fulfil this gap in the literature by studying how a lockdown policy can limit the number of infected people during an epidemic, with uncertainties on the actual number of affected individuals, and on their level of adherence to such a policy. More specially, we aim at solving this moral hazard problem by finding
\begin{enumerate}[label=$(\roman*)$]
\item the best reaction effort of the population to reduce the interaction given a specific government policy;

\vspace{-0.6em}
\item the optimal policy composed by an aggregated tax paid by the population at some fixed maturity, and a testing policy to reduce the uncertainty on the estimated number of infected people.
\end{enumerate}
This problem perfectly fits with a classical principal--agent problem with moral hazard, and boils down to finding a Stackelberg equilibrium between the principal (the leader, here the government) proposing a policy to an agent (the follower, here the population) to interact optimally in order to reduce the spread of the disease. Principal--agent problems have a long history in the economics literature, dating back from, at least, the 60s. It is not our goal here to review the whole literature on the subject, and we refer the interested reader to the seminal books by \citeauthor*{laffont2002theory} \cite{laffont2002theory}, \citeauthor*{bolton2005contract} \cite{bolton2005contract}, or \citeauthor*{salanie2005economics} \cite{salanie2005economics}. We will content ourselves to mention that this literature regained a strong momentum in the past two decades with the development of continuous-time models. Main contributors in these regards are \citeauthor*{holmstrom1987aggregation}, \cite{holmstrom1987aggregation}, \citeauthor*{schattler1993first} \cite{schattler1993first}, \citeauthor*{sannikov2008continuous} \cite{sannikov2008continuous}, see also the monograph by \citeauthor*{cvitanic2012contract} \cite{cvitanic2012contract}. 
More recently, \citeauthor*{cvitanic2017moral} \cite{cvitanic2017moral,cvitanic2018dynamic} developed a general theory allowing to tackle a great number of contract theory problem, which has been then extended and applied in many different situations\footnote{See among others \citeauthor*{aid2018optimal} \cite{aid2018optimal}, \citeauthor*{el2021optimal} \cite{el2021optimal}, \citeauthor*{cvitanic2018asset} \cite{cvitanic2018asset}, \citeauthor*{elie2019tale} \cite{elie2019tale}, \citeauthor*{elie2021mean} \cite{elie2021mean}, \citeauthor*{kharroubi2020regulation} \cite{kharroubi2020regulation}.}.
However, the previous approach requires a fundamental assumption on the structure of the controlled process, that is not satisfied in our model, because roughly speaking, there is only one Brownian motion driving the two processes, and we therefore cannot directly rely on existing result to tackle our problem. In these so-called degenerate problems, the literature has so far relied on the Pontryagin stochastic maximum principle, see for instance \citeauthor*{hu2019continuous} \cite{hu2019continuous}, but this requires extremely stringent assumptions, such as linear dynamics, which are automatically precluded for SIS/SIR models. We nevertheless prove that in our specific problem, it is possible to identity a whole family of contract representations, which is different from the (unique) one obtained in \cite{cvitanic2018dynamic}, but which still allows us to re-interpret the problem of the principal as a standard stochastic control problem. As far as we know, ours is the first paper in the literature which uses a dynamic programming approach to solve a degenerate principal--agent problem, and this constitutes our main mathematical contribution.
 
 \medskip
\textbf{Numerical results and policy-related implications.} Unfortunately, there is no way to extract from our model explicit results, especially on the shape of optimal controls. It is therefore necessary to perform numerical simulations, by implementing semi-Lagrangian schemes.
The numerical results for the SIR model are conclusive, in the sense that they confirm the relevance of a tax and testing policy to improve the control of an epidemic. First, in the benchmark case, \textit{i.e.} when the government does not put into place a specific policy, the efforts of the population are not sufficient to contain the epidemic. In our opinion, this supports the need for incentives. Indeed, if a tax policy is put into place, even in the absence of a specific testing policy, the population is encouraged to significantly reduce its interactions, thus containing the epidemic until the end of the period under consideration. Moreover, if the government also adjusts the testing policy, less effort is required on the population side, so individuals can interact almost in a business-as-usual fashion, and the epidemic is largely contained by the targeted isolation of positively-tested individuals. However, in both cases, the population relaxes its effort at the very end of the fixed lockdown period, leading to a resumption of the epidemic at that point. We obtain similar results in the case of a SIS model (see \cite[Appendix A]{hubert2020incentives}).

\medskip
{\small\textbf{Notations.} We let $\N^\star$ be the set of positive integers, $\R_+ :=[0,\infty)$ and $\R_+^\star:=(0,\infty)$. We fix a time horizon $T>0$ corresponding to the lockdown length chosen, \textit{a priori}, by the government. For every $n\in\N^\star$, $\S^n$ represents the set of $n \times n$ symmetric positive matrices with real entries. We also denote by $\Cc^n$ the space of continuous functions from $[0,T]$ into $\R^n$, and simplify notations when $n=1$ by setting $\Cc:=\Cc^1$. The set $\Cc^n$ will always be endowed with the topology associated to the uniform convergence on the compact $[0,T]$. For every finite dimensional Euclidean space $E$, and any $n\in\N^\star$, we let $\Cc_b(E,\R)$ be the space of bounded, continuous functions from $E$ to $\R$, as well as $\Cc^n_b(E,\R)$ the subset of $\Cc_b(E,\R)$ of all $n$-times continuously differentiable functions on $E$, with bounded derivatives. For every $\varphi\in \Cc^2_b(E,\R)$, we denote by $\nabla\varphi$ its gradient vector, and by $D^2\varphi$ its Hessian matrix. }

\section{Informal pandemic models and main results}

In this section, in order to highlight the results we obtain throughout this paper, we present our model in an informal way, and refer the reader to \Cref{sec:rigorous_maths} for the rigorous mathematical study. In particular, we first detail the compartmental epidemic models we consider to represent the spreading of the virus, namely a stochastic version of the well-known SIS and SIR models, and how both the population and the government can impact these dynamics. We then describe their optimal control problems, together with the Stackelberg game in which they are involved. Finally, we summarise our theoretical findings, which will prove useful for the numerical resolution described in \Cref{sec:numerical}.

\subsection{Controlled stochastic SIS/SIR dynamics}\label{ss:SIS-SIR}

At the beginning of an epidemic, it is unlikely that decision-makers, let alone the population, will have sufficient information to conclude that infected individuals become immune to the virus in question once they have recovered. This is particularly true when the virus is new, as in the case of the COVID-19. For this reason, we choose to address in our study both SIS and SIR compartmental models. The SIS model considers that infected individuals do not develop an immunity to the disease, and thus assume that an infected individual can, after recovery, re-contract the disease. Conversely, the SIR compartment model involves a third class, namely the `Recovered', \textit{i.e.}, individuals who have contracted the disease, are now cured, and especially immune to the virus under consideration. In order to make our study more comprehensive, we consider a meta-model, whose epidemic pattern is described by \Cref{fig:SIR_model}, and which allows us to deal with the two compartmental models mentioned above. We denote by $(S_t,I_t,R_t)$ the proportion of individuals in each state `Susceptible', `Infected' and `Recovered' at time $t \ge 0$. We describe below the main parameters, and whether they are controlled or not, which allows to progressively construct the final specification of the epidemic in terms of stochastic dynamics satisfied by $(S,I,R)$, given by the system \eqref{eq:SIS_SIR}.

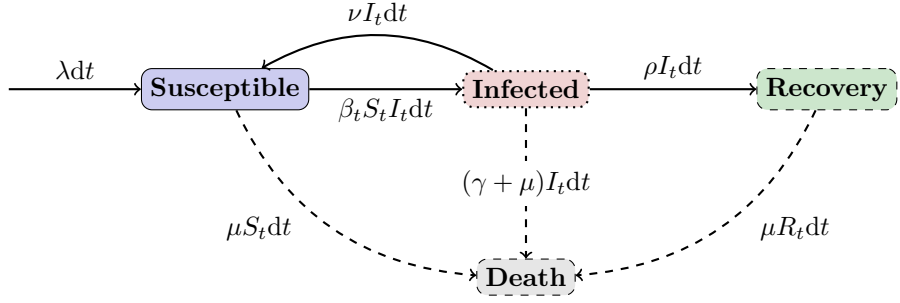
\begin{figure}[!ht]
	\centering
	\begin{tikzpicture}
		\node [rectangle] (B) at (-3,0) {};
		\node [rectangle, draw, rounded corners=4pt, fill=blue!20, text centered] (S) at (-0,0) {\bf Susceptible};
		\node [rectangle, draw, rounded corners=4pt, fill=red!20, text centered, dotted, thick] (I) at (4,0) {\bf Infected};
		\node [rectangle, draw, rounded corners=4pt, fill=gray!20, text centered, dashed] (D) at (4,-2.5) {\bf Death};
		\node [rectangle, draw, rounded corners=4pt, fill=green!20, text centered, dashed, right of=I] (R) at (7,0) {\bf Recovery};
		\draw[->, draw=black, thick]  (B) edge  node[midway, above] { \bf $\lambda \drm t$}(S);
		\draw[->, draw=black, thick] (S) edge  node[midway, below] {\bf $\beta_t S_t I_t \drm t$}(I);
		\draw[->, draw=black, dashed, thick] (I) edge  node[midway, fill=white] {\bf $(\gamma+\mu) I_t \drm t$}(D);
		\draw[->, draw=black, dashed, thick]  (S) edge[bend right]  node[midway, left] {\bf $ \mu S_t \drm t \quad$}(D);
		\draw[->, draw=black, thick]  (I) edge  node[midway, above] {\bf $ \rho I_t \drm t$ }(R);
		\draw[->, draw=black, dashed, thick]  (R) edge[midway, bend left]  node[midway, right] {\bf $\quad \mu R_t \drm t$}(D);
		\draw[->, draw=black, thick]  (I) edge[bend right]  node[midway, above] {\bf $ \nu I_t \drm t$}(S);
	\end{tikzpicture}
	\caption{SIS and SIR models with demographic dynamics}
	\label{fig:SIR_model}
\end{figure}

\medskip

\textbf{Uncontrolled parameters.} Some parameters are common in the SIS and SIR models, as highlighted by \Cref{fig:SIR_model}. In particular, both models involve three non-negative parameters: $\lambda$, $\mu$ and $\gamma$. While the parameters $\lambda$ and $\mu$ represent respectively the birth and (natural) death rates among the population, and therefore reflect the demographic dynamics unrelated to the epidemic\footnote{It should be noted that if the length of the epidemic is relatively short in relation to the life expectancy at birth, the demographic dynamics become less relevant and may be dismissed altogether, by setting $\lambda = \mu = 0$. Nevertheless, for the sake of generality, we choose to take these dynamics into account, in order to allow for a straightforward application of our study to other types of epidemics.}, while $\gamma$ represents the death rate associated to the disease. Conversely, the two non-negative constant rates $\nu$ and $\rho$ are specific to the SIS and SIR models respectively. More precisely, $\nu$ corresponds to the rate at which an infected individual returns, after recovery, to the class of susceptible individuals, while $\rho$ represents the recovery rate in the SIR model, \textit{i.e.}, the rate at which individuals who have contracted the disease are cured, and therefore immune to the virus under consideration. All the aforementioned parameters, \textit{i.e.} $\lambda,\mu,\gamma, \nu$ and $\rho$ are homogeneous to the inverse of our unit of time, \textit{i.e.} days, and are assumed to be constant and exogenous.

\medskip

\textbf{Control of the transmission rate.} The transmission rate $\beta$ of the disease is defined as the average number of contacts made by an average infective per unit of time that leads to an infection, and is therefore also homogeneous to the inverse of our unit of time, \textit{i.e.} days. In contrast to the previous parameters, $\beta$ is assumed here to be endogenous and time-dependent, in order to model the influence that the population can have on this rate. Indeed, the transmission rate of an epidemic depends essentially on two factors: the disease characteristics and the contact rate within the population. Although the population cannot modify the disease characteristics, each individual can make a costly effort to reduce his/her contact rate with other individuals in the population. With this in mind, we first assume that the constant initial transmission rate of the disease, \textit{i.e.}, without any control measures or particular effort from the population, is given by some level $\overline \beta > 0$. We then consider that the population can deviate from this initial transmission rate, namely by choosing, at some cost, a process $\beta \in \Bc$, assumed to be $B$-valued for $B:=[0,\beta^{\rm max}]$, where the constant $\beta^{\rm max} \geq \overline \beta$ represents the maximum rate of interaction that can be considered.\footnote{We refer to \Cref{ss:interaction_population} for a more precise definition of the set $\Bc$, taking into account the information flow in the model.} 

\medskip

\textbf{Compartmental model with uncertainty.} The use of a deterministic model is widespread and generally justified for most epidemics. However, when considering for example the COVID-19 pandemic, it appears that the number of infected individuals is not so simple to quantify and estimate. Indeed, without a large testing campaign, it seems complicated to know precisely the actual number of susceptible and infected, especially because of the absence of symptoms for a significant proportion of infected individuals. As a consequence, it seems more realistic for our purpose to represent the spread of the epidemic by a stochastic dynamic, which is inspired by the versions of stochastic SIS and SIR models respectively considered by \citeauthor*{gray2011stochastic} in \cite[Section 2]{gray2011stochastic} and \citeauthor*{tornatore2005stability} in \cite{tornatore2005stability}. More precisely, we consider the following dynamic for the epidemic, where the proportion of infected and susceptible are impacted at each time $t$ by a Brownian motion $W_t$
\begin{equation}\label{eq:SIS_SIR}\begin{cases}
		\displaystyle S_t=s_0+\int_0^t\big(\lambda -\mu S_s +\nu I_s-\beta_s \sqrt{\alpha_s}S_sI_s\big)\mathrm{d}s+\int_0^t\sigma \alpha_s S_s I_s \mathrm{d} W_s,\; t\in[0,T],\\[0.8em]
		\displaystyle  I_t=i_0-\int_0^t\big((\mu+\nu+\gamma+\rho) I_s-\beta_s\sqrt{\alpha_s} S_s I_s\big) \mathrm{d}s -\int_0^t\sigma \alpha_s S_s I_s \mathrm{d} W_s,\; t\in[0,T], \\[0.8em]
		\displaystyle  R_t=r_0-\int_0^t  (\rho I_s -\mu R_s) \mathrm{d}s,\; t\in[0,T],
	\end{cases}
\end{equation}
for a given initial distribution of individuals at time $0$, denoted by $(s_0,i_0,r_0) \in \R^3_+$ and assumed to be known. Note that to recover a stochastic SIS model, one has to set $\rho = 0$, and conversely $\nu = 0$ for a SIR.

\begin{remark}
	There exist several different versions of stochastic {\rm SIS/SIR} models, see among others the works by {\rm \citeauthor*{allen2008introduction} \cite{allen2008introduction} and \citeauthor*{greenwood2009stochastic} \cite{greenwood2009stochastic}}, in addition to those already mentioned in the introduction. In this paper, we assume that the uncertainty giving rise to the emergence of the Brownian motion is related to the interaction rate $\beta$. More precisely, here, $\beta$ is no longer constant compared to deterministic model but subject to random shocks, \textit{i.e.},
	$\beta \drm t \longleftarrow  \beta \drm t +\sigma \drm B_t$.
	We refer to the works by {\rm \citeauthor*{gray2011stochastic} \cite{gray2011stochastic}} and {\rm \citeauthor*{lesniewsk2004epidemic} \cite{lesniewsk2004epidemic}} for more details on the construction of such stochastic models. However, we would like to emphasise that, although we have chosen a specific dynamic, and a formulation in terms of rate and non-dimensionless groups, the general approach we develop in this paper can be adapted in a straightforward way to various stochastic models and other formulations.
\end{remark}

\textbf{Testing policy.} In addition to the parameters described above---the constant rates $\lambda$, $\mu$, $\gamma$, $\rho$ and $\nu$, and the population's control $\beta$---this stochastic version includes two new parameters: a fixed and deterministic parameter $\sigma > 0$, homogeneous to the inverse of the square root of our time unit, and a dimensionless time-dependent process $\alpha$, representing the actions of the government in terms of testing policy. More precisely, we first assume that, without any specific effort of the government, $\alpha$ is equal to $1$. Then, the government can choose to increase, at some cost, the number of tests in the population, represented by a decrease of the parameter $\alpha$, thus reducing the volatility of the processes $S$ and $I$. Hence, both the population and the government have a clearer view of the proportion of susceptible and infected, and thus on the epidemic. In particular, this control $\alpha$ of the government is assumed to be $A$-valued, where $A:=[\eps,1]$ for a small parameter $\eps\in(0,1)$\footnote{The lower bound $\eps$ is here to insist on the fact that it is not possible, or prohibitively expensive, to cancel completely the uncertainty linked to the disease's dynamics, by taking $\alpha$ to be $0$.}, and we denote by $\Ac$ the corresponding set of processes.\footnote{We refer to \Cref{ss:weak_formulation} for the rigorous definition of the set $\Ac$.}

\medskip

In addition, the testing policy allows the government to isolate positively-tested individuals. More precisely, without any testing policy, \textit{i.e.} $\alpha=1$, the government cannot isolate contaminated individuals efficiently. In this case, all infected people spread the disease, and the transmission rate of the virus is given by $\beta$. Conversely, if a testing policy is implemented by the government, \textit{i.e.} $\alpha<1$, we consider that individuals with positive test results can be isolated, and as a consequence less infected people spread the disease. In this case, the effective transmission rate is lower. We however do not assume that the impact of the testing policy on the volatility of $S$ and $I$, and on the transmission rate has the same magnitude: we expect a lower reduction of the effective transmission rate, compared to the volatility reduction for a given policy $\alpha$. Indeed, it is easier to reduce the uncertainty on the number of infected people, compared to actually isolate individuals who have been identified as infected. We thus assume a linear dependency with respect to $\alpha$ for the volatility of both $S$ and $I$, while the effective transmission rate is chosen equal to $\beta \sqrt{\alpha}$, so that the number of infected people spreading the disease at time $t$ is actually given by $\sqrt{\alpha} I_t$.

\begin{remark}
To be more realistic, the implementation of the testing policy could be modelled through the addition of a supplementary state, to capture the individuals under quarantine. The theoretical approach developed in this paper can easily be adapted to this purpose, and even for more refined compartmental models. However, the complexity of the numerical resolution increases drastically by adding a state, as mentioned in {\rm \Cref{sec:extensions}}. We therefore make the choice to limit the number of states, by considering that the testing policy has a direct impact on the effective transmission rate. Nevertheless, this shortcut should not alter the significance of our results in terms of appropriate policies, even if a more precise model would obviously give more relevant quantitative results.
\end{remark}

\subsection{The Stackelberg equilibrium}

In addition to the choice of a testing policy, the government can also incentivise the population to limit their social interactions, in order to decrease the transmission rate of the disease, by introducing financial penalties. More precisely, at time $0$, the government informs the population about its testing policy $\alpha\in\Ac$, as well as its fine policy $\chi \in \mathfrak C$\footnote{See \Cref{ss:def_contract} for a rigorous definition of the set $\mathfrak C$ of admissible fine policies.}, for the lockdown period $[0,T]$. Informally, while the testing policy directly impact the dynamic \eqref{eq:SIS_SIR} of the epidemic, the fine policy will play an indirect role: by being indexed on the proportion of susceptible and infected individuals, this tax will incentivise the population to decrease the transmission rate $\beta$, in order to limit the spread of the epidemic. Knowing these policies, the population will choose an interacting behaviour according to the following rules
\begin{enumerate}[label=$(\roman*)$]
\item an increase in the tax lowers its utility;

\vspace{-0.6em}
\item an increase in the level of interaction (up to a specific threshold, namely $\overline \beta$) improves its well-being;

\vspace{-0.6em}
\item the population is scared of having a large number of infected people.
\end{enumerate}
Then, by anticipating the optimal response of the population to a given policy $(\alpha,\chi)\in\Ac\times\mathfrak C$, the government will optimise this policy in order to maximise its own expected utility.

\subsubsection{Population optimisation problem}\label{sss:pop_problem}

For a given policy $(\alpha,\chi)\in\Ac\times\mathfrak C$, we assume that the population solves the following optimal control problem:
\begin{align}\label{pb:agent_informal}
V_0^{\rm A} (\alpha,\chi) := \sup_{\beta\in \Bc} \mathbb E \bigg[\int_0^T u(t,\beta_t, I_t)\mathrm{d}t + U(-\chi)\bigg],
\end{align}
where $u:[0,T] \times B \times \R_+ \longrightarrow \mathbb R$ and $U:\R\longrightarrow \R$ are continuous functions in all their arguments, and $U$ is a bijection from $\R$ to $\R$. Given a pair $(\alpha, \chi)$, the set of optimal contact rates $\beta$ will be denoted $\Bc^\star(\alpha, \chi)$.\footnote{Once again, the reader is referred to \Cref{ss:def_contract}, and more precisely to \Cref{eq:Bc_star} for a rigorous definition of $\Bc^\star$.}

\medskip

The functions $u$ and $U$ should be interpreted as functions translating respectively the actual value of interaction from the point of view of the population, and the disutility associated to the fine. More precisely, the function $U$ is assumed to be an increasing function, according to $(i)$ above. Concerning the function $u$, it should be non-decreasing in the second variable up to $\overline \beta$, and then non-increasing, modelling $(ii)$ above. On the other hand, the function $u$ is assumed to be non-increasing with respect to the proportion of infected individual in the population. In particular, this allows to take into account both the fear of the infection (as mentioned in $(iii)$ above) and the cost that is incurred if an individual is infected.\footnote{From the population's point of view, this cost should not actually be expressed in terms of money, but mainly corresponds to medical side effects or general morbidity. We refer to \citeauthor{anand1997disability} \cite{anand1997disability}, \citeauthor{zeckhauser1976where} \cite{zeckhauser1976where} and \citeauthor{sassi2006calculating} \cite{sassi2006calculating}, for an introduction to QALY/DALY (Quality- and Disability-Adjusted Life-Year), the generic measures of disease burden used in economic evaluation to assess the value of medical interventions.}
Moreover, we choose to normalise the utility of the population to zero when there is no epidemic. In other words, if $i_0 = 0$, then $I_t = 0$ for all $t \in [0,T]$, and thus the utility of the population should be equal to $0$. With this in mind, we assume first that $U(0) = 0$, which means that without a fine, the population does not suffer any disutility. Second, when there is no epidemic, the population should not reduce its social interaction, meaning that for all $t \in [0,T]$, $\beta_t = \overline \beta$. This leads us to assume that $u(t, \overline \beta, 0)  = 0$, for all $t \in [0,T]$.

\subsubsection{Government optimisation problem}\label{sss:gov_problem}

As already explained, the government can choose the tax $\chi \in \mathfrak C$ paid by the population at the end of the lockdown period, together with the testing policy $\alpha \in \Ac$, and we informally write its optimisation problem as
\begin{align}\label{pb:principal_informal}
    V^{\rm P}_0 := \sup_{(\alpha,\chi) \in \Xi} \sup_{\beta \in \Bc^\star (\alpha,\chi)}
    \mathbb E \bigg[\chi - \int_0^T \big( c(I_t) + k(t, \alpha_t, S_t, I_t) \big)\mathrm{d} t \bigg],
\end{align}
where $c : \R_+ \longrightarrow \R_+$ and $k : [0,T] \times A \times \R_+ \times \R_+ \longrightarrow \R$ are continuous functions. The function $c$ denotes the instantaneous cost implied by the proportion of infected people, and is thus assumed to be non-decreasing, while the function $k$ represents the cost of the testing policy. 

\medskip

In addition, the set $\Xi$ takes into account the so-called participation constraint for the population. This means that the government is benevolent, which translates into the fact that it has committed to ensure that the living conditions of the population do not fall below a minimal level. Mathematically, the government can only implement policies $(\alpha,\chi)\in\Ac\times \mathfrak C$ such that $V_0^{\rm A}(\alpha,\chi)\geq \underline v$, where the minimal utility $\underline v\in\R$ is given. This is what is encoded in the set $\Xi$.

\begin{remark}
	Recall that, while the testing policy $\alpha \in \Ac$ directly impact the dynamic of the epidemic, the tax $\chi \in \mathfrak C$ plays an indirect and incentive role. Indeed, in the moral hazard situation of interest, \textit{i.e.} when the government cannot observe the population's efforts to reduce the transmission rate of the virus, the government can only encourage the population to make efforts, by implementing an incentive scheme. In particular, by indexing the tax $\chi \in \mathfrak C$ in an optimal way on the paths of the stochastic processes $S$ and $I$, which are the only variables observable by the government in this moral hazard context, the population will be incentivised to decrease the transmission rate of the epidemic.
\end{remark}

\subsubsection{Utilities and cost specifications}\label{ex:utility_pop}

We now provide the specification for the utility and cost functions of the population and the government, respectively, that will be used for the numerical simulations in \Cref{sec:numerical}. Nevertheless, we would like to emphasise that our general approach in {\rm \Cref{sec:rigorous_maths}} does not take these specifications into account, and therefore our theoretical results are valid for very general forms of cost and utility functions. This naturally implies that alternative parameterisations could be chosen for the numerical part, if one wants to capture some costs or effects that are neglected here, for example the individual cost of being infected or the possible scaling costs of testing.

\medskip

\textbf{For the population.} Concerning the population's utility $U$ with respect to the tax $\chi$, we choose a mixed {\rm CARA}--{\rm risk-neutral} utility function, so that $U(0)=0$, and $U$ is an increasing and strictly concave bijection from $\R$ to $\R$
\[
U(x):=\frac{1-\mathrm{e}^{-\theta_{\rm p} x}}{\theta_{\rm p}}+\phi_{\rm p} x,\; x\in\R, \; \text{\rm for some } (\theta_{\rm p}, \phi_{\rm p})\in(0,+\infty)^2.
\]
For later use, we record that the inverse of $U$, denoted by $U^{(-1)}$, can be expressed in terms of the {\rm LambertW} function\footnote{See {\rm\citeauthor*{corless1996lambertw} \cite{corless1996lambertw}} for more details on the {\rm LambertW} function.}
\[
U^{(-1)}(y):=\frac1{\theta_{\rm p}}\mathrm{LambertW}\Big(\phi_{\rm p}^{-1}\mathrm{e}^{\frac{1-\theta_{\rm p} y}{\phi_{\rm p}}}\Big)+\frac{\theta_{\rm p} y-1}{\theta_{\rm p}\phi_{\rm p}},\; y\in\R.
\]
Note that the previous function $U$ defines how the population values dollars (the unit of the tax) in terms of units of utility, called \textit{util}. More precisely, \$$1$ corresponds to $U(1)$ utils, and conversely $1$ util is worth \$$U^{(-1)}(1)$.

\medskip  

Next, concerning the running utility function $u$, we can consider the following separable form
\begin{align}\label{eq:u_separable}
    u(t,b,i) &:= - u_{\beta} (t, b) - u_{\rm I} (i),\; (t,b,i)\in[0,T] \times B \times \R_+,
\end{align}
where the functions $u_\beta :[0,T] \times \R_+ \longrightarrow \R$ and $u_{\rm I}:\R_+ \longrightarrow \R$ should respectively capture the two rules $(ii)$ and $(iii)$. 
The function $u_{\rm I}$ could model the fact that the population underestimates the epidemic when the proportion of infected is close to $0$, while when it becomes large, the population is irrationally afraid. For instance, we can choose
\begin{align}\label{eq:u_I}
u_{\rm I} (i) = c_{\rm p} i^3,\; i \in \R_+,\; \text{\rm for some } c_{\rm p} \geq 0 \; \text{ (in util$\cdot$day$^{-1}$)}.
\end{align}
Finally, the function $u_\beta$ must first acknowledge that it is costly for the population to deviate from its usual contact rate. Second, during the lockdown period, the social cost of distancing measures can become more and more important for the population, and we thus expect the cost $u_\beta$ to also reflect this sensitivity with respect to time. More precisely, we can consider the following form
\begin{align}\label{eq:u_beta}
	u_{\beta} (t,b) := \eta_{\rm p} \psi(t) (\overline \beta-b)^2/2, \; (t,b)\in[0,T]\times B, \; \text{\rm for some $\eta_{\rm p} > 0$} \; \text{ (in util$\cdot$day)}.
\end{align}
Above, $\psi$ should be a non-decreasing and convex $\R_+$-valued function, to represent the increasing aversion to the lockdown for the population as time passes. In other words, deviating from its usual level of interaction entails a social cost to the population that is greater as the duration increases. More precisely, we can consider 
\[
\psi(t):= \mathrm{e}^{\tau_{\rm p} t},\; t\in[0,T], \; \text{\rm for some $\tau_{\rm p}>0$} \; \text{ (in day$^{-1}$)}.
\]

\textbf{For the government.} Regarding the cost function $c$, one can choose for instance the following linear--quadratic form $c(i):=c_{\rm g}(i+i^2)$, $i\in\R_+$, for some $c_{\rm g}$ in dollars per days, whose value is greater than $c_{\rm p}$ to take into account that the marginal cost linked to the proportion of infected people in the population is higher for the government than for the population itself.
More precisely, the linear part represents the cost per unit of infected people, while the quadratic part highlights the cost induced by the saturation of the healthcare system when the number of infected is too high. Compare to the cubic cost chosen for the population in {\rm \Cref{ex:utility_pop}}, this choice emphasises that, on the one hand, even for a small number of infected, the marginal cost faced by the government is not close to $0$ $($hence the linear term$)$. On the other hand, the population is more likely to incur very high and lasting costs in terms of QALY/DALY when the disease spreads uncontrollably, when compared to the government which mostly faces pecuniary costs.

\medskip

Concerning the cost $k$ associated with the testing policy, we may consider the following function
\begin{align*}
    k(t, a, s, i) := \dfrac{\kappa_{\rm g}}{\eta_{\rm g}} \big( a^{-\eta_{\rm g}} - 1 \big),\; (t,a,s,i)\in[0,T] \times A \times \R_+^2, \; \text{for some } \kappa_{\rm g} > 0 \; \text{ (in \$$\cdot$day$^{-1}$) and} \; \eta_{\rm g} > 0.
\end{align*}
This function highlights the fact that it is very costly, if not impossible, to eliminate the uncertainty associated with the epidemic by choosing $\alpha = 0$, while the cost of a no-testing policy $(\alpha = 1)$ is null. Indeed, on a country-wide scale, it seems impossible to develop a testing policy sufficient to know exactly the proportion of susceptible and infected.

\subsubsection{Two alternative problems}\label{sec:bench}

As already mentioned, the framework of interest in this paper is that of moral hazard, \textit{i.e.} when the government does not observe the efforts of the population, and must therefore find an optimal incentive scheme. However, in order to test the relevance of this incentive scheme, it is important to compare our results with those obtained in two more traditional settings: $(i)$ a benchmark case, when the government does not interfere, and $(ii)$ when there is no moral hazard (\textit{first-best case}), and therefore the government can enforce the optimal transmission rate on the population.

\begin{enumerate}[label=$(\roman*)$]
	\item When the government does not interfere, \textit{i.e.} without tax and testing policy, it suffices to solve \eqref{pb:agent_informal} for $\alpha=1$ and $\chi=0$. Since we assumed that $U(0) = 0$, the optimisation problem faced by the population boils down to the following standard control problem, whose associated PDE will be given by \eqref{eq:HJB_benchmark}
	\begin{equation}\label{pb:agent_benchmark}
		V_0^{\rm A}(1,0) = \sup_{\beta\in \Bc} \mathbb E\bigg[\int_0^T u(t, \beta_t, I_t)\mathrm{d}t \bigg].
	\end{equation}
	\item The first-best case is the best possible scenario where the government can enforce whichever interaction rate $\beta\in\Bc$ it desires, and simply has to satisfy the participation constraint of the population. From the practical point of view, this could correspond to a situation where the government is able to track every individual and force them to stop interacting. In this case, the problem faced by the government is 
	\begin{equation}\label{eq:FB}
		V_0^{\rm P,FB}:=\sup_{(\alpha,\chi,\beta)\in\Ac\times\mathfrak C\times \Bc} \mathbb E \bigg[\chi - \int_0^T \big( c(I_t) + k(t, \alpha_t, S_t, I_t) \big)\mathrm{d} t \bigg], \; \text{s.t.} \; \mathbb E \bigg[\int_0^T u(t,\beta_t, I_t)\mathrm{d}t + U(-\chi)\bigg] \geq \underline v.
	\end{equation}
\end{enumerate}

\subsection{Main results and comparison}\label{ss:results}

In this section, we present the main theoretical results obtained when the dynamic of the epidemic is given by \eqref{eq:SIS_SIR}: we begin by outlining the results in the two alternative problems mentioned above---the benchmark and first-best cases---then explain the optimal form of the tax in the moral hazard case, and conclude with the resolution of the government's problem in this general case. In short, the solution to any of the three problems is equivalent to solving the relevant Hamilton--Jacobi--Bellman (HJB for short) equation.

\subsubsection{The benchmark case: without tax and testing policies}\label{sec:results_bench}
As mentioned in \Cref{sec:bench} $(i)$, this benchmark problem is a standard Markovian stochastic control problem. 
In this case, the population's Hamiltonian is defined, for $t \in [0,T]$, $(s,i) \in (\R_+^\star)^2$, $p :=(p_1,p_2) \in \R^2$ and $M \in \S^2$ by
\begin{align}\label{eq:hamiltonien_benchmark}
    H^{\rm A} (t, s,i,p, M) := & \sup_{b \in B} \big\{ u(t, b, i)-b s i (p_1 - p_2)  \big\}
    + (\lambda -\mu s +\nu i) p_1
    - (\mu+\nu+\gamma+\rho) i p_2 + \frac{\sigma^2 (si)^2}2 \begin{pmatrix} 1\\ -1\end{pmatrix}^\top M \begin{pmatrix} 1\\ -1\end{pmatrix}.
\end{align}
We then have the natural identification $V_0^{\rm A}(1,0)=v(0,s_0,i_0)$, where $v$ solves the associated HJB equation
\begin{align}\label{eq:HJB_benchmark}
    - \partial_t v(t,s,i) - H^{\rm A} (t, s,i,\nabla v, D^2v) = 0, \; (t,s,i)\in\Dc,
\end{align}
with terminal condition $v(T,s,i) = 0,\; (s,i)\in\Dc_T$; where, for a particular function $F$ defined by \eqref{eq:boundary_time} in \Cref{ss:weak_formulation},
\[
\Dc := \big\{ (t,s,i) \in [0,T)\times(\R_+^\star)^2  : 0 < s+i \leq F(t,s_0,i_0) \big\},\; \Dc_T:=\big\{ (s,i)\in  \R_+^2 : 0 < s+i < F(T,s_0,i_0) \big\}.
\]

\begin{remark}
Note that if we consider a separable utility $u$, for example of the form in {\rm \Cref{ex:utility_pop}}, the maximiser of the Hamiltonian is explicitly given by $b^\circ (s,i,p_1-p_2)$, where $b^\circ$ is defined for all $(s,i,z) \in (\R_+^\star)^2 \times \R$ by
	\begin{align}\label{eq:b_circ}
		b^\circ (s,i,z) &:= \beta^{\rm max}\mathbf{1}_{\{siz<\eta_{\rm p}\psi(t)(\overline\beta-\beta^{\rm max})\}}+\bigg(\overline \beta - \dfrac{s iz}{\eta_{\rm p}\psi(t)}\bigg)\mathbf{1}_{\{\eta_{\rm p}\psi(t)(\overline\beta-\beta^{\rm max})\leq siz\leq \overline \beta \eta_{\rm p}\psi(t)\}}.
	\end{align}
In particular, the optimal interaction rate is given in this case by $\beta^\circ_t=b^\circ (S_t, I_t, (\partial_s v-\partial_iv)(t,S_t,I_t))$, $t\in [0,T]$.
\end{remark}

\subsubsection{The first-best case: without moral hazard}

To find the optimal interaction rate $\beta \in \Bc$, as well as the optimal contract $(\alpha, \chi) \in \Ac \times \mathfrak C$, in the first-best case, one has to solve the government's problem defined by \eqref{eq:FB}. Mathematical details are postponed to \Cref{sss:fb_maths}, but we present here an overview of the main results. To take into account the participation constraint, one has to introduce the associated Lagrangian. Given a Lagrange multiplier $\varpi>0$, we first remark that the optimal tax is constant and given by $\chi^\star(\varpi):=-\big(U^\prime\big)^{(-1)}\big(1 / \varpi\big).$ Then, defining for any $\varpi>0$
\begin{align}\label{eq:V0_varpi_FB}
    \overline V_0(\varpi):=\sup_{(\alpha,\beta)\in\Ac\times\Bc}\E\bigg[\int_0^T\big(\varpi u(t,\beta_t,I_t)-c(I_t)-k(t,\alpha_t,S_t,I_t)\big)\mathrm{d}t\bigg],
\end{align}
we have
\begin{align}
\label{eq:V0_P_FB}
V_0^{\rm P,FB} =
\inf_{\varpi>0}\Big\{ \chi^\star(\varpi)+\varpi \big(U\big(-\chi^\star(\varpi)\big)-\underline v\big)
+ \overline V_0(\varpi)
\Big\}.
\end{align}
Note that $\overline V_0(\varpi)$ is the value function of a standard stochastic control problem, and therefore we expect to have $\overline V_0(\varpi)= v^\varpi(0,s_0,i_0)$, for a function $v^\varpi:[0,T]\times\R_+^2 \longrightarrow \R$ solution to the following HJB PDE
\[
\displaystyle
-\partial_t v^\varpi(t,s,i) + c(i) 
- (\lambda-\mu s+\nu i) \partial_s v^\varpi
+ (\mu+\nu+\gamma+\rho)i \partial_i v^\varpi
- \Hc^{\varpi} (t,s,i,\partial v^\varpi, D^2 v^\varpi) =0,\; (t,s,i)\in \Dc,
\]
with terminal condition $v^\varpi(T,s,i)=0,\; (s,i)\in\Dc_T$, where the Hamiltonian is defined, for $t \in [0,T]$, $(s,i) \in (\R_+^\star)^2$, $p :=(p_1,p_2) \in \R^2$ and $M \in \S^2$ by
\begin{align*}
    \Hc^{\varpi} (t,s,i,p,M) := \sup_{a \in A} \bigg\{
\sup_{b \in B} \big\{
\varpi u(t,b,i) - b si \sqrt{a}  (p_1-p_2)
\big\} - k(t,a,s,i) 
+ \dfrac12 \sigma^2 (si)^2 a^2 (M_{11} -2 M_{12} + M_{22})
\bigg\}.
\end{align*}

\begin{remark}
	If we consider for instance the utilities given in {\rm\Cref{ex:utility_pop}} for the utility $u$, the optimal interaction rate is given for all $t \in [0,T]$ by $\beta^{\varpi}_t = b^{\varpi} \big(S_t, I_t, \partial v^\varpi (t,S_t,I_t), \alpha_t \big)$, for $\alpha \in \Ac$ and a Lagrange multiplier $\varpi > 0$, where $b^{\varpi} (s,i,p,a) := b^\circ \big( s, i,\sqrt{a} (p_1-p_2) / \varpi \big), \; \text{for all} \; 
		(s,i,p,a) \in (\R_+^\star)^2 \times \R^2 \times A,$
	recalling that $b^\circ$ is defined by \eqref{eq:b_circ}.
\end{remark}

\subsubsection{Relevant form of tax policy}\label{sss:relevant_contract}

Let us now return to the main problem, \textit{i.e.} the case with moral hazard. One of the main theoretical result of our study is given by \Cref{thm:agent}. Informally, this theorem states that given an admissible contract, namely a testing policy $\alpha \in \Ac$ and a tax $\chi \in \mathfrak C$, there exist a unique $Y_0$ and $Z$ such that the following representation holds
\begin{align}\label{eq:Uxi_representation}
	U(-\chi) = Y_0
	- \int_0^T \Big(Z_t (\mu+\nu+\gamma+\rho) I_t +  u(t, \beta^\star_t, I_t) - \beta^\star_t\sqrt{\alpha_t} S_t I_t Z_t \Big) \drm t
	- \int_0^T Z_t \drm I_t,
\end{align}
where $\beta^\star$ is the unique optimal contact rate for the population. More precisely, we can state that for (Lebesgue--almost every) $t \in [0,T]$, $\beta^\star_t := b^\star (t, S_t, I_t, Z_t)$ is the maximiser of the function $b \in B \longmapsto u(t,b,I_t) - b S_t I_t Z_t$. Under some assumptions for existence and smoothness of the inverse of the function $U$, the previous equation naturally gives a representation for the tax $\chi$. Based on \eqref{eq:Uxi_representation}, the tax $\chi$ will be indexed on the variation of the proportion of infected $I$, through the stochastic integral $\int_0^\cdot Z_s\mathrm{d}I_s$, and not on the variation of susceptible $S$ $($though it is indexed on $S$ through the $\mathrm{d}t$ integral$)$. Nevertheless, using the link between the dynamics of $I$ and $S$, we can write \begin{align}\label{eq:Uxi_representation2}
	U(-\chi) =  Y_0
	- \int_0^T \big(u(t, \beta^\star_t, I_t) - \beta^\star_t\sqrt{\alpha_t} S_t I_t Z_t - Z_t (\lambda -\mu S_s +\nu I_s ) \big) \drm t
	+ \int_0^T Z_t \drm S_t.
\end{align}
Through this equation, we can state that the tax can alternatively be indexed on $S$ instead of $I$. Therefore, given the strong link between the number of Susceptible and the number of Infected, it is sufficient to index the tax on only one of these two quantities, and one can therefore choose indifferently to index the tax $\chi$ on the variations of $I$ or $S$. The reader familiar with contract theory in continuous-time will have noticed that the previous representation for the tax $\chi$ is not exactly the expected one. Indeed, referring for instance to \citeauthor*{cvitanic2018dynamic} \cite{cvitanic2018dynamic} the contract is usually the sum of three components: a constant similar to $Y_0$, chosen by the Principal in order to satisfy the participation constraint of the Agent; an integral with respect to time $t \in [0,T]$ of the agent's Hamiltonian; a stochastic integral with respect to the controlled process, \textit{i.e.}, in our framework, $(S,I)$. Neither the representation \eqref{eq:Uxi_representation} nor \eqref{eq:Uxi_representation2} are, {\it a priori} of this form. This difference is due to the fact that the dynamics of $(S,I)$ is degenerated. More precisely, there is a fundamental structure condition in \cite{cvitanic2018dynamic} requiring that the drift of the output process belongs to the range of its volatility. In words, defining for $ (s,i) \in \R_+^2$ and $(a,b) \in A \times B$,
\begin{align*}
	\sigma(i,s,a):=\sigma asi
	\begin{pmatrix}
		1\\
		-1
	\end{pmatrix}, \text{ and }  
	\lambda(s,i,b,a):=\begin{pmatrix}
		\lambda-\mu s+\nu i+b \sqrt{a} si\\
		-(\mu+\nu+\gamma+\rho)i+b \sqrt{a}si
	\end{pmatrix},
\end{align*}
the condition assumed in \cite[Equation (2.1)]{cvitanic2018dynamic} implies that 
$\lambda(s,i,b,a)\propto\sigma(i,s,a)$, for any $(s,i,a,b)\in\R_+^2\times A\times B$, which is obviously impossible here. Therefore, we cannot use directly any existing result in the literature, and we should not expect, {\it a priori}, to be able to obtain a contract representation similar to the one in \cite{cvitanic2018dynamic}, nor that the so-called dynamic programming approach will prove effective in our case. Indeed, as far as we know, such degenerate models have only been tackled using the stochastic maximum principle, see \citeauthor*{hu2019continuous} \cite{hu2019continuous}. However, and somewhat surprisingly, the form we exhibit for the tax is actually strongly related to the usual representation. The reason for this is twofold. First, up to the sign, the volatilities in the dynamics of both $S$ and $I$ are exactly the same. Second, both the processes $S$ and $I$ are driven by the same Brownian motion $W$. Therefore, intuitively, in order to provide incentives to the population, the government can afford to index the tax on only one of the two processes. Mathematically, it is also straightforward to show that given an arbitrary decomposition of the process $Z$ in \Cref{eq:Uxi_representation} of the form $Z=:Z^s-Z^i$, we have exactly the general form provided in \cite{cvitanic2018dynamic}. The main difference is that in \cite{cvitanic2018dynamic}, $Z^s$ and $Z^i$ are both uniquely given, while in our representation, only their difference actually matters. Hence, there is an infinite number of possible representations for the tax $\chi$ in our degenerate model.

\subsubsection{Government's problem in the general case}

Thanks to the reasoning developed in \Cref{sec:rigorous_maths}, we are able to determine the optimal design of the fine policy and the associated optimal effort of the population. In particular, as informally explained in the previous section, to implement a tax policy $\chi \in \mathfrak C$, the government only needs to choose a constant $Y_0$ and a process $Z$. Given these two parameters, we can state that the optimal contact rate for the population is defined by $\beta^\star_t := b^\star(t, S_t, I_t, Z_t,\alpha_t)$, such that the function $b \in B \longmapsto u(t,b,I_t) - b \sqrt{\alpha_t} S_t I_t Z_t$ is maximised for (Lebesgue--almost every) $t \in [0,T]$.\footnote{If we consider separable utilities, as in {\rm\Cref{ex:utility_pop}}, the maximiser $b^\star$ is given for all $(t,s,i,z,a) \in [0,T] \times (\R_+^\star)^2 \times \R\times A$ by $b^\star (s,i,z,a) := b^\circ (s, i, z \sqrt{a})$, recalling that $b^\circ$ is defined by \eqref{eq:b_circ}.} It thus remains to solve the government's problem in order to determine the optimal choice of $Y_0$ and $Z$. The reader is referred to \Cref{ss:solving_gov_pb} for the rigorous government's problem, but, to summarise the results, the optimal process $Z$ as well as the optimal testing policy $\alpha$ are determined so as to maximise the government's Hamiltonian, given by
\begin{align*}
    H^{\rm P} (t, s, i, p, M)
    = &\ \sup_{z \in \R,a\in A} \bigg\{
  b^\star(t, s,i,z,a)\sqrt{a} s i (p_2 - p_1)+
     \dfrac{1}{2} \sigma^2 a^2 (si)^2 f(z,M) 
    -  k(t, a, s, i)    - u^\star (t, s,i,z,a) p_3\bigg\} \\[0.3em]
    &+ (\lambda - \mu s +\nu i) p_1 
    - (\mu+\nu+\gamma+\rho) i p_2
    - c (i),
\end{align*}
for $(t,s,i,p,M)\in[0,T]\times\R_+^2\times\R^3\times\S^3$, and where, in addition for $z \in \R$,
\[
f(z,M) := M_{11} - 2 M_{12} + M_{22} - 2 z (M_{23} -M_{13}) + z^2 M_{33},\; \text{\rm and}\; u^\star (t, s,i,z,a) := u \big(t, b^\star(t, s,i,z,a) ,i\big).
\]
Finally, it remains to solve numerically the following HJB equation, for all $t \in [0,T]$ and $x := (s,i,y) \in \R^3$
\begin{align}
\label{eq:HJBGen}
    - \partial_t v (t, x)  - H^{\rm P} \big(t, x, \nabla_x v, D^2_x v\big) = 0,\; (t,x)\in \Oc,\; v(T,x)= -U^{(-1)}(y),\; x\in\Oc_T,
\end{align}
where the natural domain over which the above PDE must be solved is
\[
\Oc:=\big\{ (t,s,i,y) \in [0,T) \times \R_+^2 \times \R: 0 < s+i < F(t,s_0,i_0) \big\},\; \Oc_T:=\big\{ (s,i,y) \in \R_+^2 \times \R: 0 < s+i < F(T,s_0,i_0) \big\}.
\]

\section{Numerical experiments}\label{sec:numerical}

The results presented in \Cref{ss:results} are quite theoretical: except for the optimal transmission rate, it is complicated to obtain explicit formulae for the other variables sought, in particular for the optimal testing policy $\alpha$, even if we consider separable utility functions as in \Cref{ex:utility_pop}. It is therefore necessary to perform numerical simulations to evaluate the optimal efforts of the population and the government, as well as the optimal tax policy. Given the similarities in the results between the SIS and SIR models, only those related to the SIR model are presented in this section. The reader will find in \cite{hubert2020incentives} the results corresponding to the SIS model.

\subsection{Choice of parameters}

The set of parameters used for the simulations of the epidemic dynamics given by \eqref{eq:SIS_SIR} are provided in \Cref{tab:params_SIR} and are inspired by those chosen by \citeauthor*{elie2020contact} \cite{elie2020contact}. Recall that the parameter $\overline \beta$ denotes the usual contact rate within the population, before the beginning of the lockdown. In other words, $\overline \beta$ represents the initial and effective transmission rate of the disease, without any specific effort of the population. The associated reproduction number $\mathcal R_0$, commonly defined by $\mathcal R_0 := \overline \beta / (\nu+\rho)$ in the literature on epidemic models, is equal to $2.0$, and is thus in the confidence interval of available data, see for example \citeauthor{li2020early} \cite{li2020early}. Then, the parameters $\lambda$ and $\mu$ represent respectively the birth and (natural) death rates among the population, and therefore reflect the demographic dynamics unrelated to the epidemic, while $\gamma$ represents the death rate associated to the disease. To simplify, and since the duration of the COVID-19 epidemic should be relatively short in comparison to the life expectancy at birth, we choose to disregard the demographic dynamics by setting $\lambda = \mu = 0$. In contrast, we set $\gamma = 1\%$, since the mortality associated with the disease appears to be significant. Finally, recall that the parameters $\nu$ and $\rho$ correspond respectively to the recovery rates in the SIS and SIR models. Since we want to consider here a SIR dynamic, we let $\nu = 0$ and $\rho = 0.1$, to account for the average 10-day duration of COVID-19 disease.

\renewcommand{\arraystretch}{1.5}
\begin{table}[!ht]
\centering
\small
\begin{tabular}{|c|c|c|c|c|c|c|c|c|c|c|}
\hline
    $T$ (days) & $(s_0,i_0, r_0)$  & $(\lambda, \mu)$& $\gamma$& $\nu$& $\rho$& $\sigma$& $\overline \beta$\\
    \hline
    $200$ & $(0.99984, 1.07 \times 10^{-4}, 5.3 \times 10^{-5})$ & $(0,0)$ & $0.01$  & $0$ & $0.1$ & $0.1$
    & $0.2$  \\ \hline
\end{tabular}
\caption{\small Set of parameters for the simulations of SIR model.}
\label{tab:params_SIR}
\end{table}

In addition, the following numerical experiments are implemented using the utility and cost functions mentioned in \Cref{ex:utility_pop}.
These functions require to specify several parameters, provided in \Cref{tab:params_contract}. 

\begin{table}[H]
	\small  \centering
	\begin{subtable}[b]{.45\linewidth}
		\centering
		\begin{tabular}{|c|c|c|c|c|c|c|}
			\hline
			Parameters & $c_{\rm p}$ & $\eta_{\rm p}$ & $\theta_{\rm p}$ & $\tau_{\rm p}$ & $\phi_{\rm p}$ & $\beta^{\rm max}$ \\
			\hline
			Values & $0.5$ & $1$ & $4$ & $0$ & $0.5$ & $0.2$  \\
			\hline
		\end{tabular}
		\caption{\small Characteristics of the population.}
		\label{tab:params_contract_pop}
	\end{subtable}
	\begin{subtable}[b]{.45\linewidth}
		\small    \centering
		\begin{tabular}{|c|c|c|c|c|}
			\hline
			Parameters & $\kappa_{\rm g}$ & $c_{\rm g}$ & $\eta_{\rm g}$ & $\varepsilon$ \\
			\hline
			Values & $0.001$ & $1$ & $0.01$ & $0.01$ \\
			\hline
		\end{tabular}
		\caption{\small Characteristics of the government.}
		\label{tab:params_contract_gov}
	\end{subtable}\vspace{-0.7em}
	\caption{\small Set of parameters for cost and utility functions}\label{tab:params_contract}
\end{table}

When not explicitly specified, the simulations presented in this section are performed with the sets of parameters described in \Cref{tab:params_contract,tab:params_SIR}. 
However, the parameters used to describe in particular the utility and cost functions of the population and government are set in a relatively arbitrary way. 
To actually estimate these parameters would require an extensive sociological and economic study, that we do not presume to be able to perform at this stage, and linking, for example, the population's costs to the DALY/QALY concepts already mentioned, and the government's costs to those of the health care system and its possible congestion.
Moreover, there is considerable uncertainty in the literature on the choice of all parameters used to describe the dynamics of the epidemic, in particular because the COVID-19 is a new type of virus.
It will therefore be necessary to study the sensitivity of the results obtained with respect to the selected parameters.

\medskip

Finally, it should be remembered that, in contrast to usual principal--agent problems, the government implements a mandatory tax, which the population cannot refuse. Nevertheless, we consider that the government is benevolent, in the sense that it still wishes to ensure that the utility of the population remains above a certain level, denoted by $\underline v$. To fix this level, we assume that the government wants to ensure at the very least to the population the same living conditions it would have had in the event of an uncontrolled epidemic, \textit{i.e.}, without any effort on the part of neither the population nor the government, meaning $\beta= \overline \beta$, $\alpha = 1$ and $\chi = 0$.
Mathematically, this is equivalent to the following, since $u$ is separable of the form \eqref{eq:u_separable}, such that for all $t \in [0,T]$, $u_{\beta} (t, \overline \beta) =0$ and $u_{\rm I}$ satisfies \eqref{eq:u_I}
\begin{align}\label{eq:reservation_utility}
\underline v 
:= \mathbb \E^{\P^{1,\overline\beta}} \bigg[\int_0^T u(t, \overline\beta, I_t) \mathrm{d}t + U(0)\bigg]
= \mathbb \E^{\P^{1,\overline\beta}} \bigg[- \int_0^T \big( u_{\rm I} (I_t) + u_{\beta} (t, \overline\beta) \big) \mathrm{d}t \bigg]
= - c_{\rm p} \mathbb \E^{\P^{1,\overline\beta}} \bigg[\int_0^T I_t^3 \mathrm{d}t \bigg].
\end{align}
Notice that the reservation utility $\underline v$ is given by the worst case scenario, without any sanitary precaution neither from the population nor from the government. This level may be judged too severe, and one could consider a model where the government is more benevolent. Nevertheless, the value of $\underline v$ should not be of major importance, since it should only impact the initial value $Y_0$.

\subsection{Numerical approach}\label{ss:numerical_approach}

In order to solve \Cref{eq:HJB_benchmark} corresponding to the population's problem in the benchmark case, as well as \Cref{eq:HJBGen} for the government's problem, we need a method permitting to deal with degenerate HJB equations. We choose to implement semi-Lagrangian schemes, first proposed in \citeauthor{camilli1995approximation} \cite{camilli1995approximation}. 
These are explicit schemes using a given time-step $\Delta t$, and requiring interpolation on the grid of points where the equation is solved.
This interpolation can be either linear, as proposed in \cite{camilli1995approximation}, or using some truncated higher-order interpolators, as proposed by \citeauthor{warin2016some} \cite{warin2016some}, leading to convergence of the numerical solution to the viscosity solution of the problem. A key point here, which makes the approach delicate, is that the domain over which the PDEs are solved is unbounded. It is therefore necessary to  define a so-called resolution domain, over which the numerical solution will be actually computed, which on the one hand must be large enough, and which on the other hand creates additional difficulties in the treatment of newly introduced boundary conditions. In order to treat these issues, we use two special tricks:
\begin{enumerate}[label=$(\roman*)$]
    \item picking randomly the control in \eqref{eq:SIS_SIR} for the benchmark case, and in \eqref{eq:edsGen} for the general case, and using the forward SDE with an Euler scheme, a Monte Carlo method allows us to get an envelop of the reachable domain with a high probability.
    Then, given a  discretisation step, the grid of points used by the semi-Lagrangian scheme is defined at each time-step with bounds set by the reachable domain estimated by Monte Carlo.
    Therefore, at time $0$, the grid is represented by one mesh, while their number can reach millions near $T$;
    
    \vspace{-0.6em}
    \item since the scheme is explicit, starting at $t$, it requires to use only some discretisation points at date $t+\Delta t$, and a modification of the scheme is implemented to use only points inside the grid at date $t+ \Delta t$, as shown in \cite{warin2016some}.
\end{enumerate}
Lastly, in dimension 3 or above, parallelisation techniques defined in \cite{warin2016some} have to be used. The numerical results below are obtained using the StOpt library, see \citeauthor*{gevret2018stochastic} \cite{gevret2018stochastic}.

\subsection{The benchmark case}\label{sec:benchSIR}

We first focus on the benchmark case, when the government does not implement any particular policy to tackle the epidemic, \textit{i.e.}, $\alpha = 1$ and $\chi = 0$. Recall that in this case, the population's problem is given by \eqref{pb:agent_benchmark}, and is then equivalent to solving the HJB equation \eqref{eq:HJB_benchmark}. For our simulations, we choose a number of time-steps equal to $200$, and a discretisation step equal to $0.0025$. The interpolator is chosen linear, and the optimal command $b^\circ$ used to maximise the Hamiltonian is discretised with $200$ points given a step discretisation of $0.005$.
Once the PDE is solved, a forward Euler scheme is used to obtain trajectories of the optimally controlled $S$ and $I$, meaning with the optimal transmission rate $b^\circ$. In order to check the accuracy of the method described in \Cref{ss:numerical_approach}, we implement two versions of the resolution: the first version is a direct resolution of 
\eqref{eq:HJB_benchmark} with the Hamiltonian \eqref{eq:hamiltonien_benchmark}; the second one relies on a change of variable. More precisely, we consider $(s,x:=(s+i))$ as state variables, instead of $(s,i)$, and then solve the problem \eqref{eq:HJB_benchmark}, but with a slightly modified Hamiltonian to take into account this change of variable. The advantage of the second representation is that the dispersion of $I_t+S_t$ is zero and thus smaller than the one of $I_t$, leading to the use of grids with a smaller number of points. First, to give an overview of the overall trend, we plot, on \Cref{fig:SIR_Simu_benchmark}, $100$ trajectories of the optimal interaction rate $\beta^\star$, and the associated proportions $S_t$ and $I_t$ of susceptible and infected, using the resolution method $(i)$ mentioned above, \textit{i.e.}, with state variables $(S,I)$. For more accurate trajectories, we compare on \Cref{fig:SIRBenchASim} two different trajectories of the optimal interaction rate $\beta^\star$, together with the corresponding dynamic of the proportion $I$ of infected. For these two simulations, we compare the results given by the two aforementioned methods. More precisely, while the blue curve is obtained through the direct resolution, the orange one results from the second method, \textit{i.e.}, with state variables $(S,S-I)$. Finally, on \Cref{fig:SIRSimusTau}, we test the influence of the parameter $\tau_{\rm p}$ by setting $\tau_{\rm p} = 0.01$, instead of $0$.
\begin{figure}[H]
\centering
   \begin{minipage}[c]{.29\linewidth}
          \includegraphics[width=\linewidth]{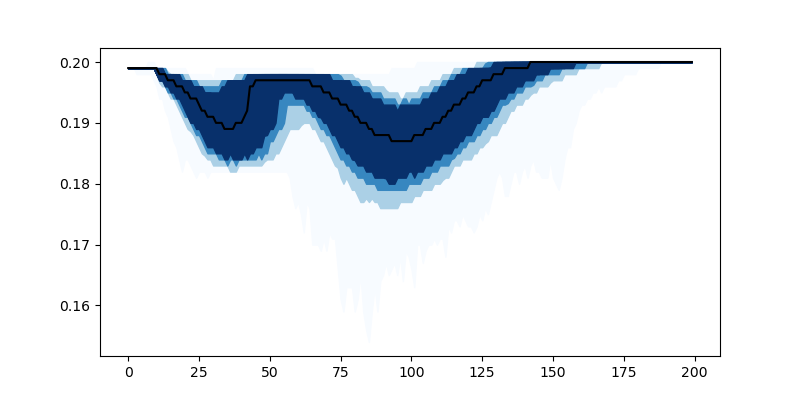}\vspace{-1em}
          \caption*{\small Optimal effort $\beta^\star$}
   \end{minipage}
   \begin{minipage}[c]{.29\linewidth}
      \includegraphics[width=\linewidth]{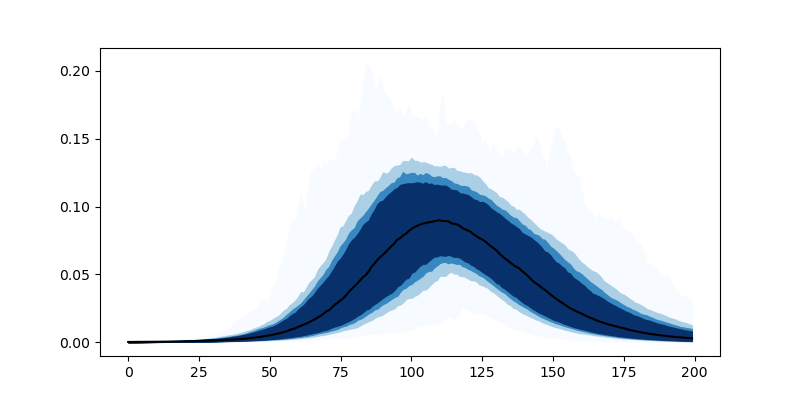}\vspace{-1em}
      \caption*{\small Proportion $I$ of infected}
   \end{minipage}
      \begin{minipage}[c]{.29\linewidth}
      \includegraphics[width=\linewidth]{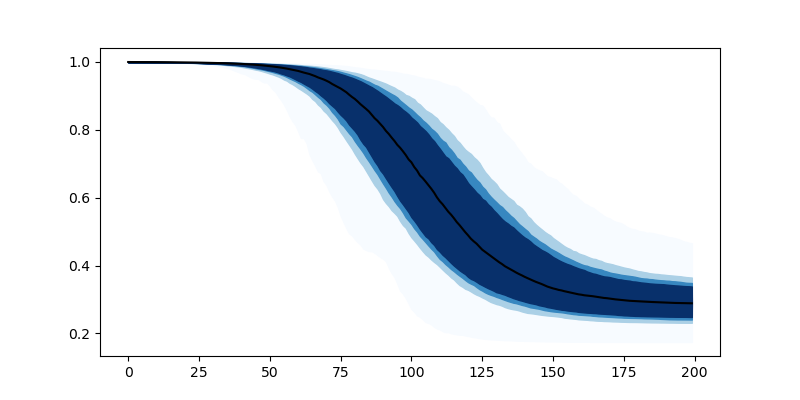}\vspace{-1em}
       \caption*{\small Proportion $S$ of susceptible}
   \end{minipage}\vspace{-0.7em}
\caption{\small Dispersion of $1000$ simulations with respect to time of the SIR model in the benchmark case.}\label{fig:SIR_Simu_benchmark}
\end{figure}

\textbf{Voluntary lockdown of the population.} As expected, the optimal behaviour $\beta^\star$ is to start close to $\overline \beta$, then to decreases as the disease spreads in the population. More specifically, two waves of effort can be observed: the first one delays the acceleration of the epidemic, and the second, generally more significant, takes place during the peak of the epidemic.
Approaching the fixed maturity, individuals come back to their usual behaviour $\overline \beta$. 

\begin{figure}[H]
\centering
\begin{minipage}[c]{.29\linewidth}
    \centering
    \includegraphics[width=\linewidth]{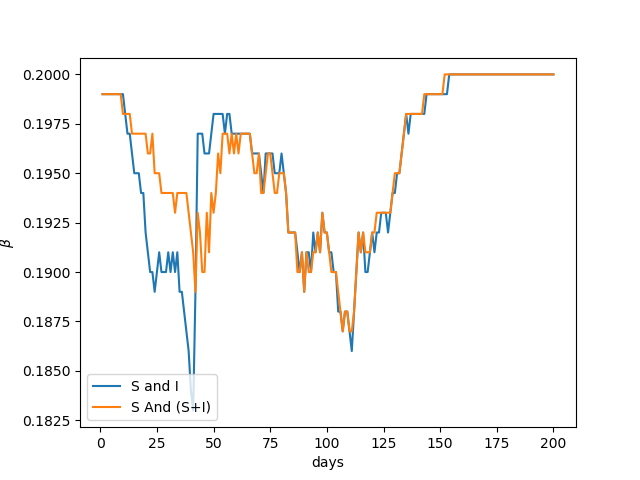} \vspace{-2em}
    \caption*{\small Proportion $I$ of infected}
\end{minipage}
\begin{minipage}[c]{.29\linewidth}
    \centering
    \includegraphics[width=\linewidth]{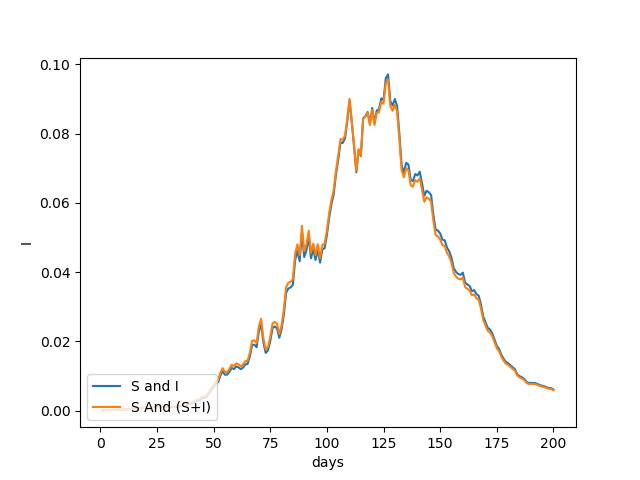} \vspace{-2em}
    \caption*{\small Optimal effort}
\end{minipage}   \vspace{-0.7em}
\caption{\small The optimal transmission rate $\beta$ and the resulting proportion $I$ in the benchmark case}\label{fig:SIRBenchASim}
{\small Comparison between of the two methods aforementioned on two simulations.}
\end{figure}

\textbf{Sensitivity with respect to the method.}
As we can see in \Cref{fig:SIRBenchASim} (top), the optimal effort exhibits the same features as those previously described. Moreover, the blue curve and the orange curve, representing respectively the results of the two aforementioned methods, are very close, except at the beginning of the time interval, probably because of the very small initial value $i_0$.
Nevertheless, we can see that the two methods lead to the same dynamic for the proportion of infected, since the two curves are almost superposed. Therefore, a small error on the computation of the optimal effort at the beginning does not impact the optimally controlled trajectories of $I$.
The resolution with respect to $(s,s+i)$ seems to be more regular, and may give a command  closer to the analytical one.

\medskip
\textbf{The fear of the infection is not enough.} Without a proper government policy to encourage the lockdown, the natural reduction of the interaction rate among individuals is not sufficient to contain the disease, so that it spreads with a high infection peak, up to $0.175$. As a result, even if at the end of the time interval under consideration, the epidemic appears to be over, between $60$ and $80\%$ of the population has been contaminated by the virus, since the proportion $S$ at time $T=200$ lies between $0.2$ and $0.4$. In conclusion, without some governmental measures, the fear of the epidemic is not sufficient to encourage the population to make sufficient effort, in order to significantly reduce the rate of transmission of the disease. The introduction by the government of an effective lockdown policy together with an active testing policy should improve the results of the benchmark case, in particular by reducing the peak of infection and the total number of infected people over the considered period.

\medskip
\textbf{The lockdown fatigue.} By setting $\tau_{\rm p} = 0.01$ instead of $0$, the cost of the lockdown from the population's point of view is now increasing with time. This allows to take into account the possible fatigue the population may suffer if the lockdown continues for too long. As expected, by comparing \Cref{fig:SIRSimusTau,fig:SIR_Simu_benchmark}, the impatience of the population gives higher values of optimal interaction rate $\beta$. Moreover, we can see that the second wave of effort is more impacted (\textit{i.e.}, the contact rate is less reduced) by the impatience of the population than the first one.

\begin{figure}[H]
\centering
   \begin{minipage}[c]{.29\linewidth}
          \includegraphics[width=\linewidth]{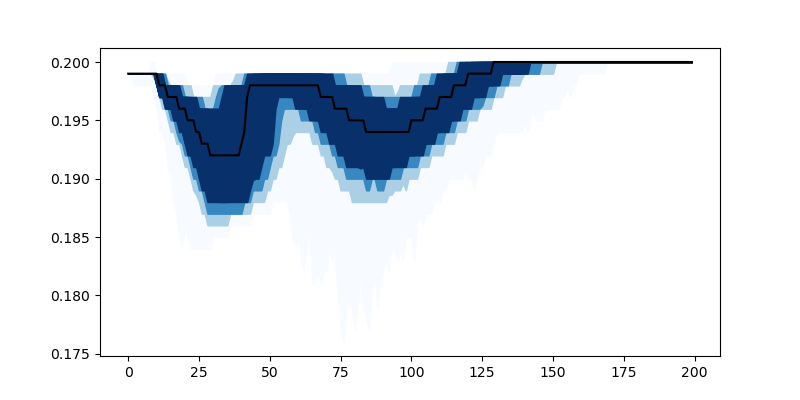}\vspace{-1em}
          \caption*{\small Optimal effort $\beta^\star$}
   \end{minipage}
   \begin{minipage}[c]{.29\linewidth}
      \includegraphics[width=\linewidth]{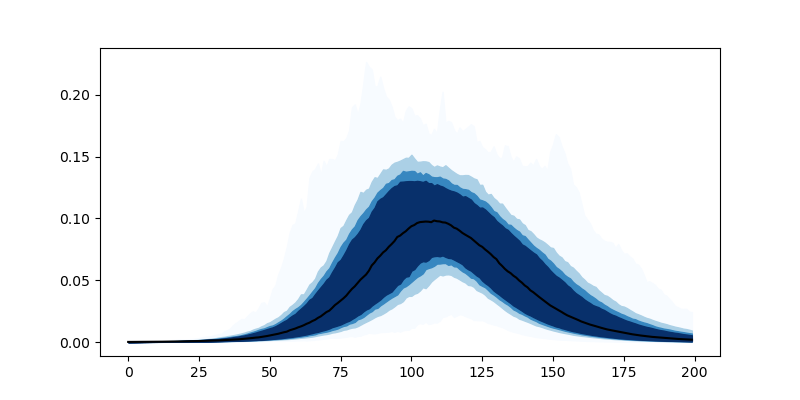}\vspace{-1em}
      \caption*{\small Proportion $I$ of infected}
   \end{minipage}
      \begin{minipage}[c]{.29\linewidth}
      \includegraphics[width=\linewidth]{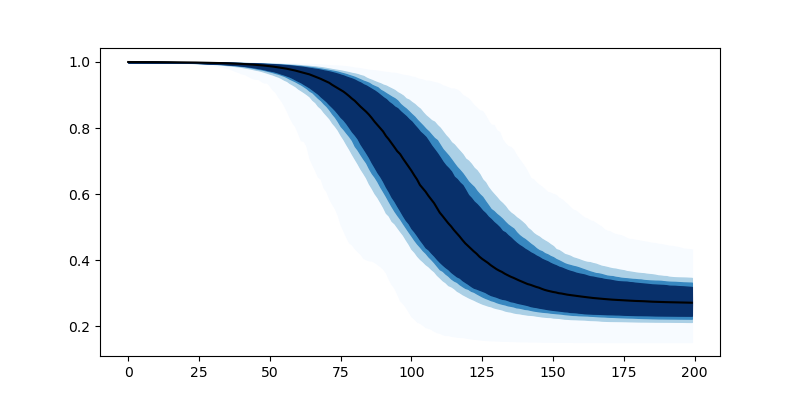}\vspace{-1em}
       \caption*{\small Proportion $S$ of susceptible}
   \end{minipage}\vspace{-0.7em}
	\caption{\small Dispersion of  simulations of the SIR model in the benchmark case with $\tau_{\rm p} =0.01$}\label{fig:SIRSimusTau}
\end{figure}
\subsection{Lockdown policy, without testing}

We focus in this section on the tax policy, by assuming that $A = \{1\}$. In such a situation, \textit{i.e.}, without a proper testing policy, the detection and hence the isolation of ill people becomes very intricate. This case is interesting, as it corresponds to the lockdown policy that most of western countries have implemented in 2020, when faced with the COVID-19 disease, while a very small number of tests was available. 
Indeed, most countries put in place systems of fines, or even prison sentences, to incentivise people to lockdown. Although the penalties for non-compliance are not as sophisticated as in our model, most governments did adapt the level of penalties according to the stage of the epidemic: higher fines during periods of strict lockdown (hence at the peak of the epidemic), or in case of recidivism, for example. This reflects the adjustment of sanctions in many countries according to the health situation, and therefore a notion of dynamic adaptation to circumstances, which is exactly what is suggested by our tax system. 
Though it is clear that our model is different from reality, since  in most countries, the fine is paid by a particular individual who has not complied with the injunctions, we still believe it allows to highlight sensible guidelines.

\medskip
The numerical approach is highly similar to the method used to solve the benchmark case. One difference is that we have to estimate the reservation utility of the population, namely $\underline v$, given by \eqref{eq:reservation_utility}. Using a Monte Carlo method and a Euler scheme with a time-discretisation of 200 time-steps and $10^6$ trajectories, we obtain an approximated value $\underline v =-0.02937$.
Then, we can solve \eqref{eq:HJBGen} through the aforementioned semi-Lagrangian scheme, with $200$ time steps, as well as a step discretisation for the grid in $(s,i,y)$ corresponding to $(0.0025,0.0025,0.005)$, leading to a number of meshes at maturity equal to $250 \times 70  \times 800$. A last technical point concerning the domain of the control $Z$. Although this control of the government, used to index the tax on the proportion of infected, can take high values, we have to bound its domain in order to perform the numerical simulations. We choose to restrict its domain to an interval $[-Z_{\rm max}, Z_{\rm max}]$, and consider a discretisation step equal to $0.5$.
One would naturally expect that a larger choice would lead to somewhat better solutions. However, this neglects a fundamental numerical issue: large values of $Z$ increase the numerical cost, as they enlarge the volatility of the process $Y$ (given by $\sigma Z I S$). As such, since the volatility cone becomes larger, it is necessary to sample a much larger grid in order to be able to cover the region were $Y$ will most likely take its values. Too large values of $Z_{\rm max}$ therefore become numerically intractable, unless one is willing to sacrifice accuracy. A balance need to be struck, which is why we capped $Z_{\rm maz}$ at $30$. 

\medskip
First, we present in \Cref{fig:optISIR} different trajectories of the proportion $I$ of infected when the government implements the optimal tax policy, and compare it to the trajectories obtained in the benchmark case. As mentioned before, we also want to study the sensibility with respect to the arbitrary bound $Z_{\rm max}$, and we thus represent the paths of $I$ in three cases, in addition to the benchmark case: for $Z_{\rm max} = 10$ (orange curves), $Z_{\rm max} = 20$ (green), and $Z_{\rm max} = 30$ (red).
Then, the corresponding simulations of the optimal control $Z$ of the government, used to index the tax on the proportion of infected, are given in \Cref{fig:optZSIR}. We compare optimal controls $\beta$ and $Z$ for the tax policy with different lockdown time period in \Cref{fig:matSIR}.
Finally, \Cref{fig:optBetaSIR} regroups the simulations of the optimal transmission rate $\beta^\star$ obtained with the tax policy, and compare it to $\beta^\circ$ obtained in the benchmark case. 

\begin{figure}[H]
\centering
    \begin{minipage}[c]{.29\linewidth}
          \includegraphics[width=\linewidth]{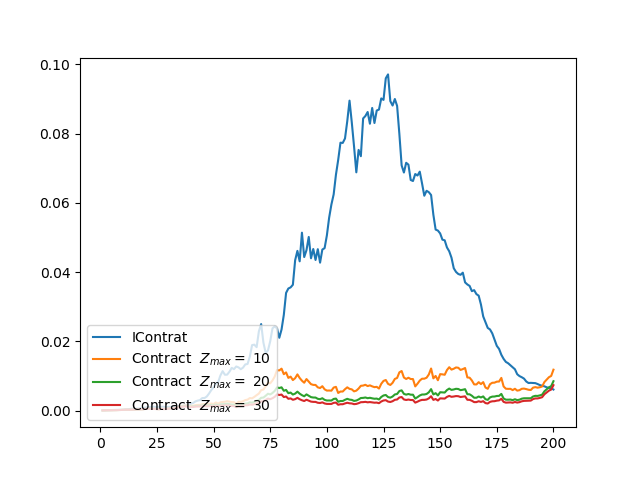}\vspace{-1em}
          \caption*{\small \textit{Simulation $1$}}
      \end{minipage}
    \begin{minipage}[c]{.29\linewidth}
          \includegraphics[width=\linewidth]{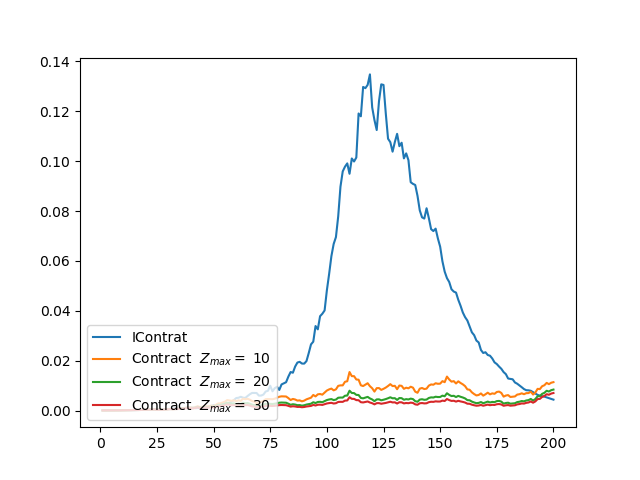}\vspace{-1em}
          \caption*{\small \textit{Simulation $2$}}
      \end{minipage}
        \begin{minipage}[c]{.29\linewidth}
          \includegraphics[width=\linewidth]{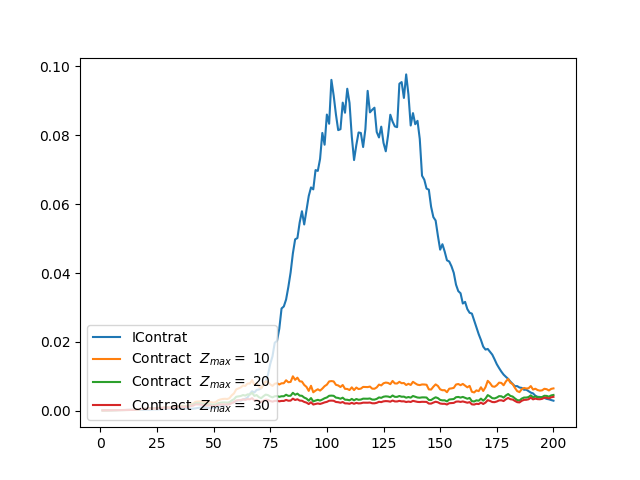}\vspace{-1em}
          \caption*{\small \textit{Simulation $3$}}
      \end{minipage}\vspace{-0.7em}
      \caption{\small Optimal trajectories of $I$ without testing. Comparison for different values of $Z_{\rm max}$ and for the benchmark.}\label{fig:optISIR} 
\end{figure}

\textbf{The epidemic is at best contained, and at worst delayed.} Compared to the benchmark case, we observe in \Cref{fig:optISIR} that the optimal lockdown policy prevents the epidemic peak by maintaining low levels of infection during the lockdown period. Therefore, the government has more time to prepare for a possible infection peak after the lockdown, specifically to increase hospital capacity and provide safety equipment (surgical masks, hydro-alcoholic gel, respirators...). The government can also use this time to fund the development of tests to detect the virus, as well as the research on a vaccine or a remedy for the related disease. However, we can see that at the end of the lockdown period, in many cases the virus is not exterminated and the epidemic may even restart. This is also illustrated by \Cref{fig:IContratSIR}, representing the dispersion of $500$ trajectories of $I$, obtained with the optimal control. Such a phenomenon can be understood as follows: the lockdown slows down the epidemic, so that a very small proportion of the population has been infected and is therefore immune. We thus cannot rely on herd immunity, which is reached here if at least 50\% of the population has been contaminated, to prevent a resurgence of the epidemic. Consequently, this lockdown policy is a powerful leverage to delay an epidemic, but this tool needs to be supplemented by alternative policies. If the time saved through lockdown is not exploited, it will have no impact on the final consequences of the epidemic.

\medskip
\textbf{Policy implications.} We first remark in \Cref{fig:optZSIR} that the shape of the optimal indexation parameter rate $Z$ remains the same, regardless of the simulation and the value of $Z_{\rm max}$. More importantly, we will see that the paths of the optimal transmission rate associated to different $Z_{\rm max}$, are almost superposed. As a consequence, and as previously exhibited in \Cref{fig:optISIR}, the value of $Z_{\rm max}$ has a minor impact on the trajectories of $I$ itself. On the shape of the control $Z$, we remark that it first takes the most negative value possible ($-Z_{\rm max}$) for about $20$ days, then increases almost instantaneously to reach the maximum value $Z_{\rm max}$, before slowly decreasing to $0$. Therefore, the optimal tax scheme set by the government is as follows. 
First, at the beginning of the epidemic, it seems optimal to give to the population a compensation as high as possible, by setting $Z = -Z_{\rm max}$. Though this may be a numerical artefact, the fact that this appeared in all our simulations tends to show that it is actually significant. We interpret this as the government anticipating the negative consequences of the lockdown policy by immediately providing monetary relief to the population. This is exactly what happened in several countries, for instance in the USA with stimulus checks sent to every citizen, and our model endogenously reproduces this aspect. Policy-wise, it shows that maximum efficiency for such packages is attained when they are provided to the population as early as possible.

\begin{figure}[H]
	\begin{minipage}[c]{.29\linewidth}
		\includegraphics[width=\linewidth]{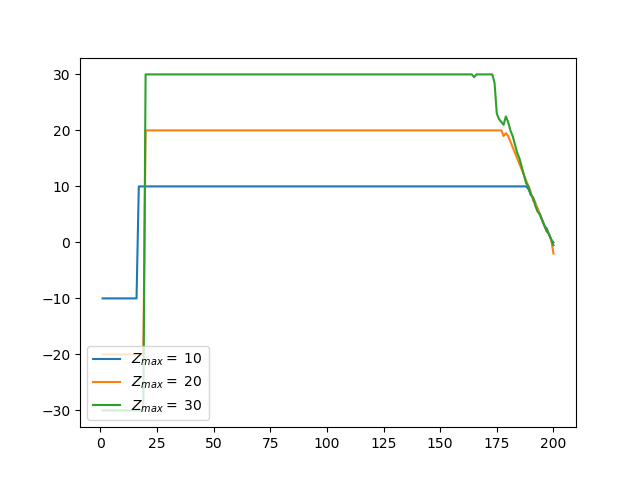}\vspace{-1em}
		\caption*{\small\textit{ Simulation $1$}}
	\end{minipage}   
	\begin{minipage}[c]{.29\linewidth}
		\includegraphics[width=\linewidth]{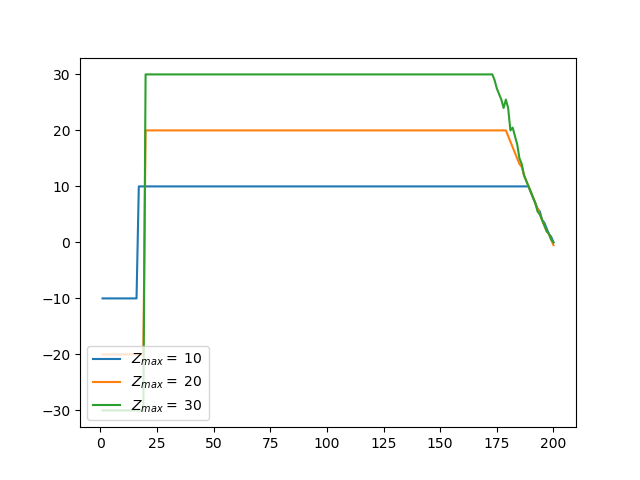}\vspace{-1em}
		\caption*{\small \textit{Simulation $2$}}
	\end{minipage}  
	\begin{minipage}[c]{.29\linewidth}
		\includegraphics[width=\linewidth]{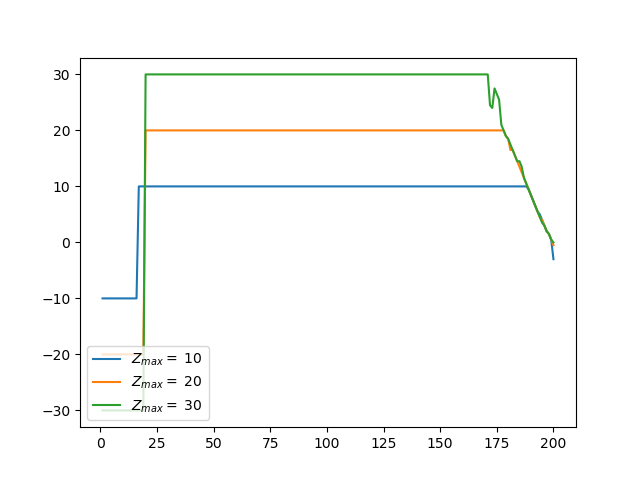}\vspace{-1em}
		\caption*{\small \textit{Simulation $3$}}
	\end{minipage} \vspace{-0.7em}
	\centering
	\caption{\small Optimal trajectories of the control $Z$ without testing. Comparison for different values of $Z_{\rm max}$, with $A =\{1\}$.} \label{fig:optZSIR} 
\end{figure}

Approaching the maturity, the government eases the lockdown. However, this may be premature, since we have observed in the previous figures that the epidemic may restart at the end of the period. Indeed, considering a final time horizon is equivalent to assuming that `the world' stops at that time: costs generated by the epidemic after $T$ are not taken into account. Nevertheless, this boundary effect has no impact on the previous results and interpretations. Indeed, we remark that if we consider a more distant time $T$, the lockdown certainly lasts longer, but follows the exact same patterns (see \Cref{fig:matSIR} below). 
Moreover, the lockdown period should still end at some time, which is why a finite terminal time is assumed. This time may correspond to an estimate of the time needed to implement other more sustainable policies, such as the implementation of an active testing policy, or to wait for the discovery of a vaccine.

\begin{figure}[H]
    \begin{minipage}[c]{.29\linewidth}
    \centering
          \includegraphics[width=\linewidth]{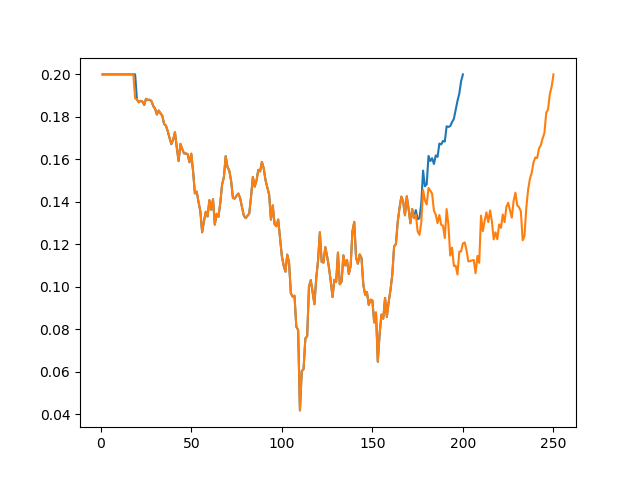}\vspace{-1em}
          \caption*{\small Optimal $\beta$}
      \end{minipage}  
    \begin{minipage}[c]{.29\linewidth}
    \centering
        \includegraphics[width=\linewidth]{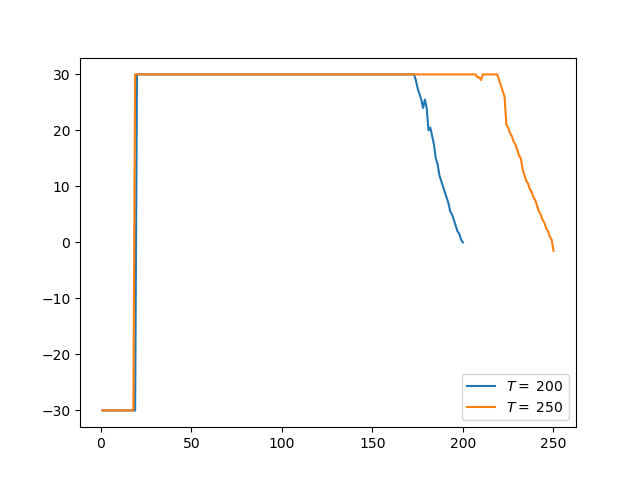}\vspace{-1em}
          \caption*{\small Optimal control $Z$}
      \end{minipage}   \vspace{-0.7em}
\centering
\caption{\small Maturity effect for the tax policy in the SIR model}\label{fig:matSIR}
{\small Comparison of the optimal trajectories of $Z$ for $T=200$ and $T=250$, with $Z_{\rm max}=30$.}
\end{figure}

\textbf{Optimal tax sensitivity with respect to the lockdown duration.}
On \Cref{fig:matSIR}, we give two trajectories of the optimal contact rate $\beta$ and the optimal $Z$ for two different maturities. It is clear that both trajectories follow the same paths until some point. Regardless of the maturity, the contact rate $\beta$ and the parameter $Z$ have the same characteristics as those shown respectively in \Cref{fig:optZSIR,fig:optBetaSIR}. As one approaches the shortest maturity, \textit{i.e.} $T=200$, the parameter $Z$ decreases towards $0$, while the other remains at the maximum, and decreases later. Therefore, the fact that $Z$ decreases at maturity, as mentioned above, appears to be a boundary effect.

\begin{figure}[H]
	\centering 
	\begin{minipage}[c]{.29\linewidth}
		\includegraphics[width=\linewidth]{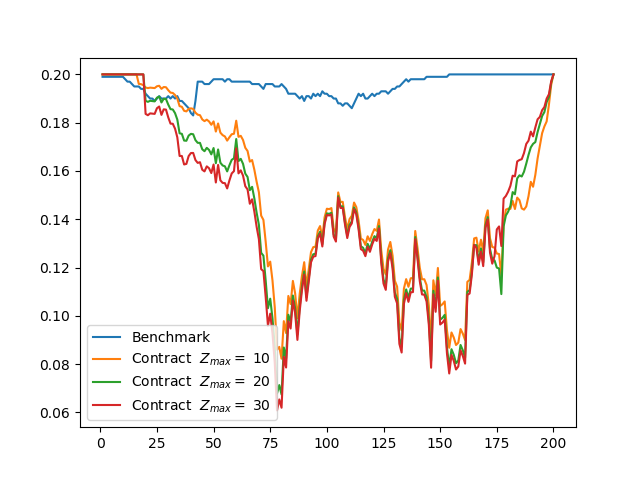}\vspace{-1em}
		\caption*{\small \textit{Simulation $1$}}
	\end{minipage}   
	\begin{minipage}[c]{.29\linewidth}
		\includegraphics[width=\linewidth]{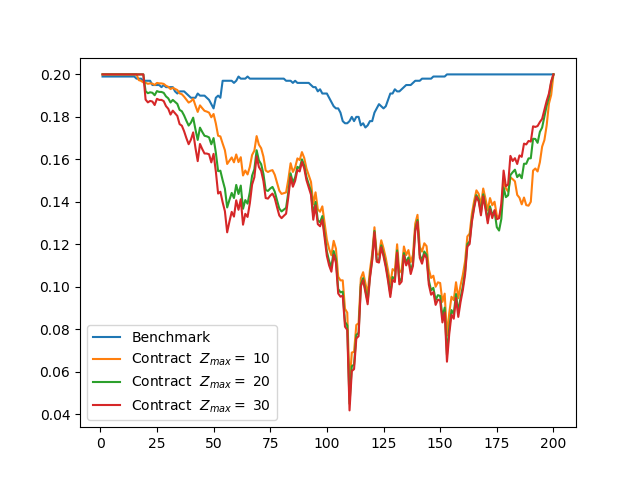}\vspace{-1em}
		\caption*{\small \textit{Simulation $2$}}
	\end{minipage}   
	\begin{minipage}[c]{.29\linewidth}
		\includegraphics[width=\linewidth]{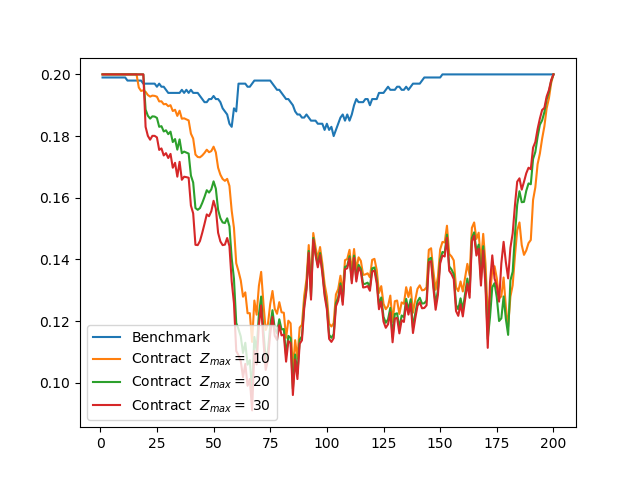}\vspace{-1em}
		\caption*{\small \textit{Simulation $3$}}
	\end{minipage}   \vspace{-0.7em}
	\caption{\small Optimal transmission rate $\beta$ without testing}
	{\small Comparison for different $Z_{\rm max}$ and with the benchmark case, in the case $A=\{1\}$.}
	\label{fig:optBetaSIR}
\end{figure}

\textbf{Optimal interaction rate and comparison with the benchmark case.} In the beginning, recall that $Z$ is negative, meaning that the tax is negatively indexed on the variation of $I$. In other words, since $I$ is globally (but very slightly) increasing at the beginning of the epidemic, the compensation increases with $I$, which means that the population is not incentivised at all to decrease their contact rate, and thus the transmission rate of the virus, which remains equal to the initial level $\overline \beta$. Then, as the epidemic spreads, $Z$ becomes very high, which now incentivises the population to reduce the transmission rate below $\overline \beta$. Finally, near the end of the lockdown period, $Z$ plunges to zero, which naturally implies that the optimal contact rate $\beta^{\star}$ goes back to its usual level $\overline \beta$.
\begin{figure}[H]
\centering
\begin{minipage}[c]{.3\linewidth}
\centering
\includegraphics[width=\linewidth]{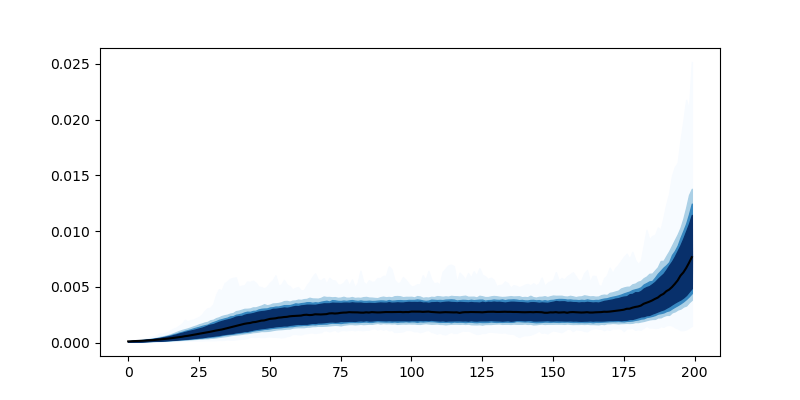}\vspace{-1em}
\caption*{\small Contract case}
\end{minipage}  
\begin{minipage}[c]{.3\linewidth}
\centering
\includegraphics[width=\linewidth]{SIRCovidSLDump25BarBeta2I0107SInit99984Beta200ndt200Pol1tau0rho10etap10cP5050T200_ITraj.png}\vspace{-1em}
\caption*{\small Benchmark case}
\end{minipage} \vspace{-0.7em}
\caption{Dispersion of  simulations of the proportion $I$ of infected in the SIR model}\label{fig:IContratSIR} 
{\small Comparison between the case with tax policy (but without testing) on the left and the benchmark case on the right.}
\end{figure}

\subsection{Tax policy with testing}

In this section, we now study the case where the government can implement an active testing policy, in addition to the incentive policy for lockdown, to contain the spread of the epidemic.
This policy is similar to the one adopted by most European governments in June 2020, after relatively strict containment periods and at a time when the COVID-19 epidemic seemed to be under control. Indeed, the lockdown periods in Europe have generally made it possible to delay the epidemic, and thus to give public authorities time to prepare a meaningful testing policy. This has two major interests. First, it allows the identification of clusters, and therefore provides a more precise knowledge of the dynamics of the epidemic in real time. Second, by identifying infected people, we can force them to remain isolated. Thus, by developing a robust testing policy, public authorities can in fact relax the lockdown while keeping the rate of disease transmission at a sufficiently low level. Therefore, comparing with the no-testing policy case, we expect that
\begin{enumerate}[label=$(\roman*)$]
    \item the government will be able to control the epidemic at least as well as with just the lockdown policy;
    \item it will allow the population to regain a contact rate closer to the desired and initial level $\overline \beta$.
\end{enumerate}

To study the optimal testing policy $\alpha^\star$, taking values in $A := [\varepsilon, 1]$, we consider the cost of effort $k$ given in \Cref{ex:utility_pop}. This cost function emphasises the fact that testing the entire population every day is inconceivable, and therefore results in an explosion of cost when $\alpha$ takes values close to $0$. Recall that the parameters for the function $k$, namely $\kappa_{\rm g}$ and $\eta_{\rm g}$ are given in \Cref{tab:params_contract_gov}. 
Finally, $A$ is discretised with a step equal to $0.05$ and we consider $Z_{\rm max}=30$.

\begin{figure}[H]
\centering
    \begin{minipage}[c]{.29\linewidth}
          \includegraphics[width=0.99\linewidth]{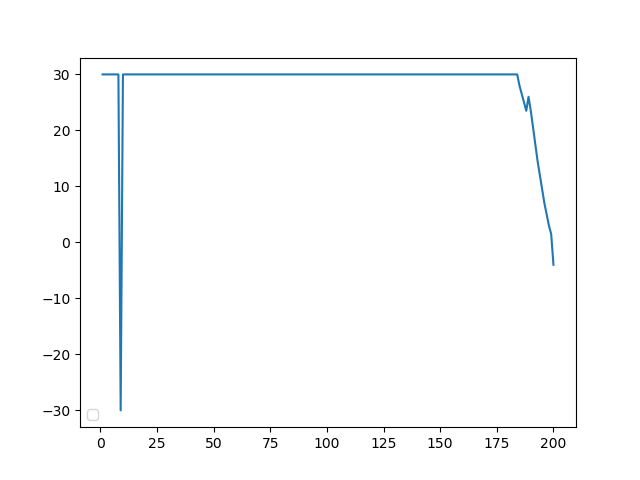}\vspace{-1em}
          \caption*{\small  \textit{Simulation $1$}}
      \end{minipage}   
    \begin{minipage}[c]{.29\linewidth}
          \includegraphics[width=0.99\linewidth]{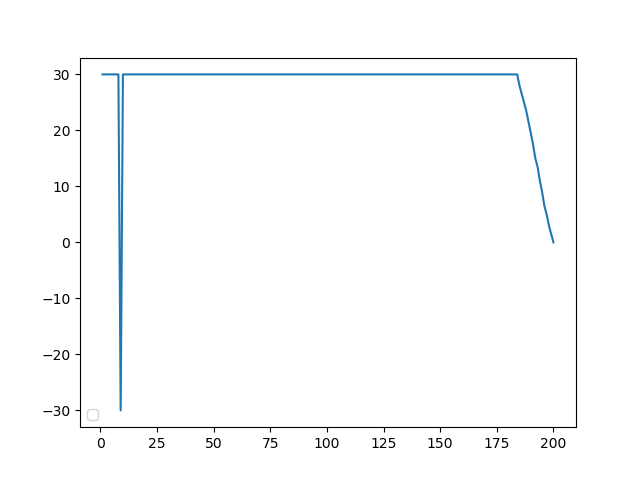}\vspace{-1em}
          \caption*{\small  \textit{Simulation $2$}}
      \end{minipage}   
        \begin{minipage}[c]{.29\linewidth}
          \includegraphics[width=0.95\linewidth]{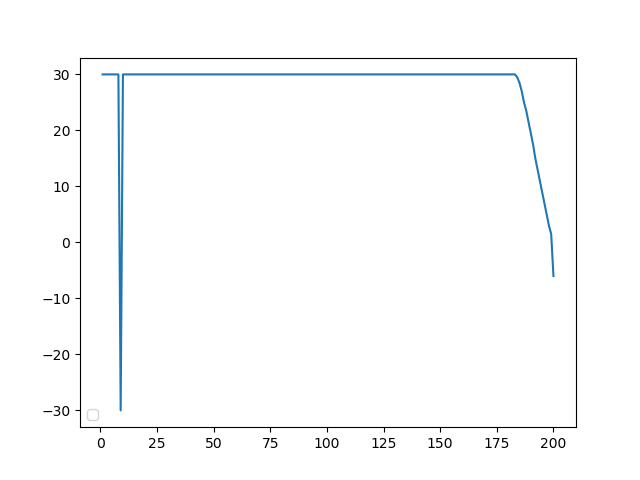}\vspace{-1em}
          \caption*{\small  \textit{Simulation $3$}}
      \end{minipage}   \vspace{-0.7em}
      \caption{\small Optimal trajectories of $Z$ with testing policy.}\label{fig:optZSIRGen}
\end{figure}
\begin{figure}[H]
\centering
    \begin{minipage}[c]{.29\linewidth}
          \includegraphics[width=0.99\linewidth]{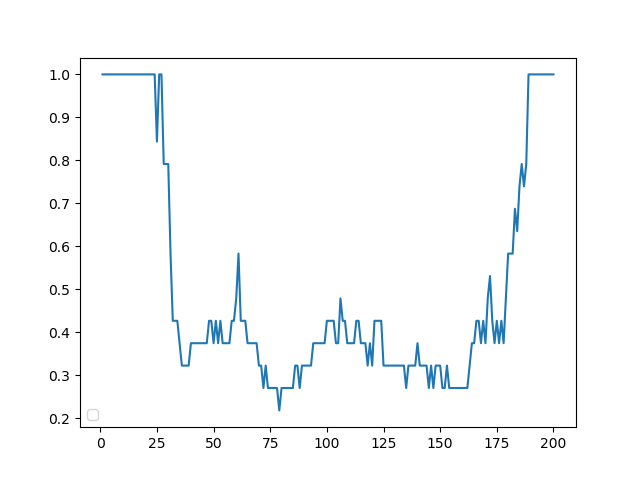}\vspace{-1em}
          \caption*{\small  \textit{Simulation $1$}}
      \end{minipage} 
    \begin{minipage}[c]{.29\linewidth}
          \includegraphics[width=0.99\linewidth]{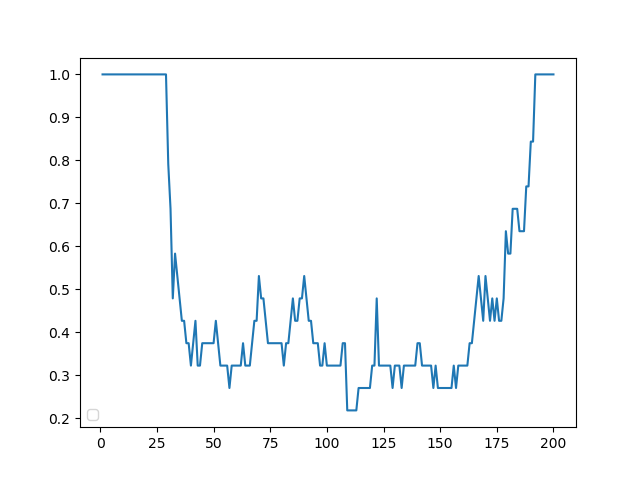}\vspace{-1em}
          \caption*{\small  \textit{Simulation $2$}}
      \end{minipage}  
        \begin{minipage}[c]{.29\linewidth}
          \includegraphics[width=0.95\linewidth]{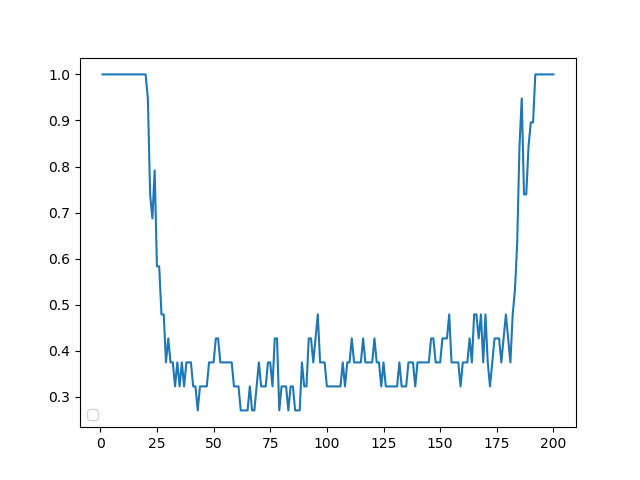}\vspace{-1em}
          \caption*{\small  \textit{Simulation $3$}}
      \end{minipage}   \vspace{-0.7em}
      \caption{ \small Optimal trajectories of the testing policy $\alpha$}\label{fig:optASIRGen}
\end{figure}
\begin{figure}[H]
\centering
    \begin{minipage}[c]{.29\linewidth}
          \includegraphics[width=0.99\linewidth]{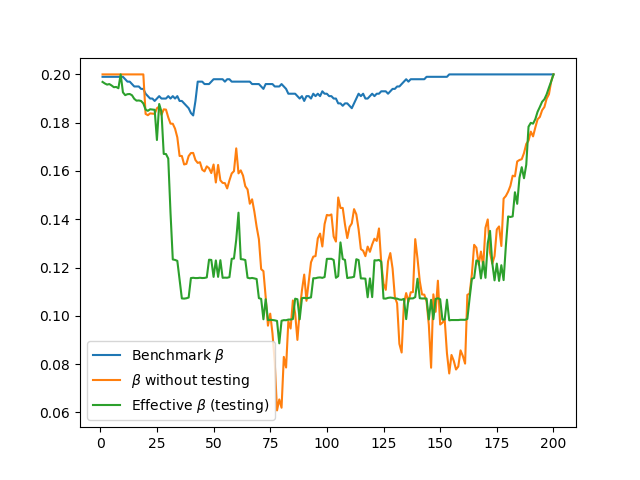}\vspace{-1em}
          \caption*{\small  \textit{Simulation $1$}}
      \end{minipage}  
    \begin{minipage}[c]{.29\linewidth}
          \includegraphics[width=0.99\linewidth]{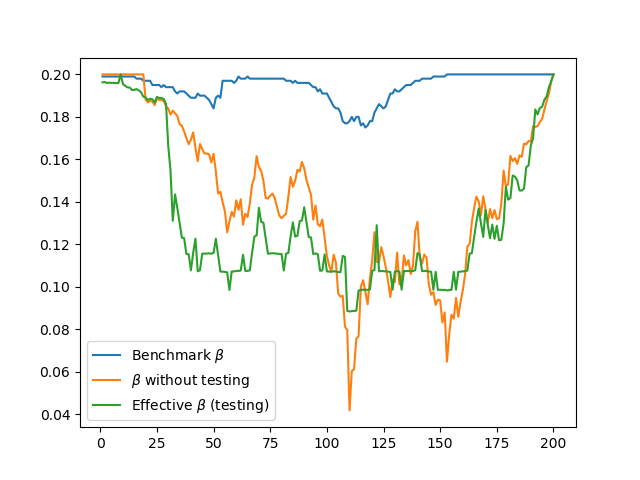}\vspace{-1em}
          \caption*{\small \textit{Simulation $2$}}
      \end{minipage}   
        \begin{minipage}[c]{.29\linewidth}
          \includegraphics[width=0.95\linewidth]{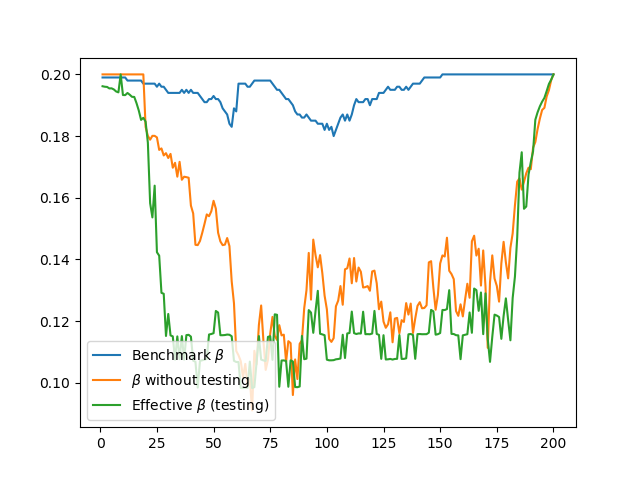}\vspace{-1em}
          \caption*{\small  \textit{Simulation $3$}}
      \end{minipage}   \vspace{-0.7em}
\caption{\small Optimal effective transmission rate $\beta \sqrt{\alpha}$ with testing policy}\label{fig:optBetaSIRGen} 
{\small Comparison between the three cases, the benchmark, with, and without testing.}
\end{figure}

\textbf{Relaxed lockdown but lower effective transmission rate.} 
First, comparing \Cref{fig:optZSIRGen,fig:optZSIR}, the optimal control $Z$ presents the same shape in both cases, except at the beginning, since now $Z$ is not negative initially. In fact, we observe that the government is asking for less effort from the population, and therefore the \textit{initial stimulus} mentioned in the paragraph `Policy implications' still happens, but later and for a much shorter length.
\Cref{fig:betaTestingSIR} also shows that the optimal contact rate is closer to the initial level $\overline \beta$, which should induce a more violent spread of the disease.
Nevertheless, the control $\alpha$, representing the testing policy and given by \Cref{fig:optASIRGen}, balances this effect. 
Indeed, the testing allows an isolation of targeted infected individual, and therefore contribute to the decrease of the effective transmission rate of the disease, represented in \Cref{fig:optBetaSIRGen}. 
Therefore, 
comparing \Cref{fig:SIRGen} with \Cref{fig:IContratSIR}, we notice that the control of the epidemic is more efficient than in the case $A=\{ 1 \}$, since the proportion of infected is globally decreased. 
Finally, \Cref{fig:SIRGen} gives a global overview with the dispersion of $500$ simulations for the optimal controls $\alpha$ and $Z$ as well as for the proportion $I$ of infected, which confirms the intuition given by the three selected ones.

\begin{figure}[ht]
\centering
	\begin{minipage}[c]{.29\linewidth}
		\centering
		\includegraphics[width=\linewidth]{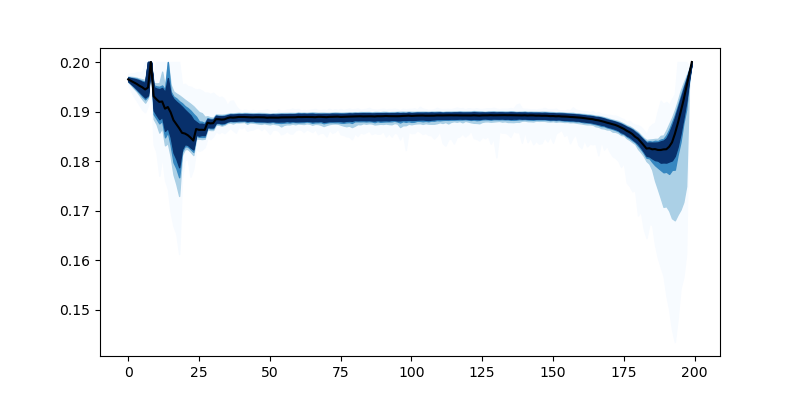}
		\caption*{\small Optimal contact rate $\beta$}
	\end{minipage} 	
	\begin{minipage}[c]{.29\linewidth}
		\includegraphics[width=\linewidth]{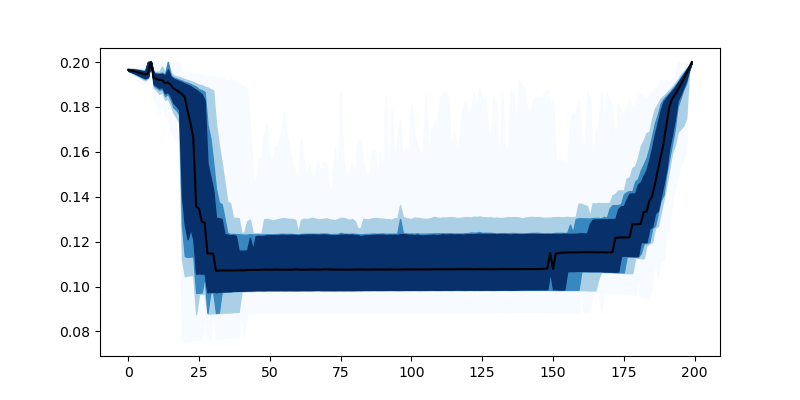}
		\caption*{\small Effective transmission rate $\beta \sqrt{\alpha}$}
	\end{minipage}  
	\caption{\small Dispersion of  simulations of the transmission rate with testing policy}\label{fig:betaTestingSIR}
\end{figure}

\begin{figure}[H]
\centering
\begin{minipage}[c]{.29\linewidth}
\includegraphics[width=0.99\linewidth]{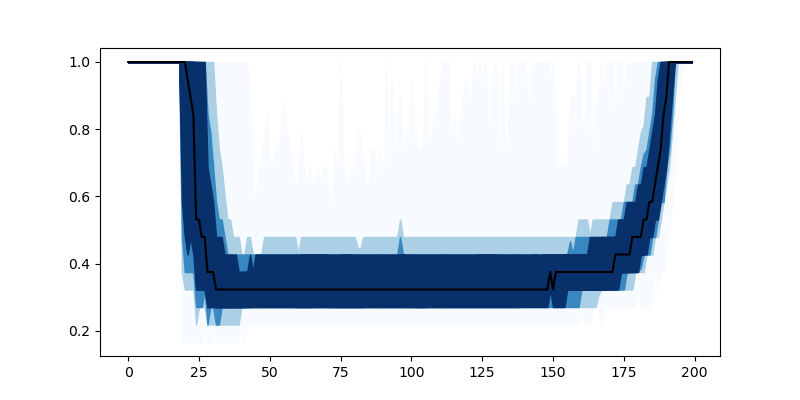}
\caption*{\small Optimal control $\alpha$}
\end{minipage}  
 \begin{minipage}[c]{.29\linewidth}
  \includegraphics[width=0.99\linewidth]{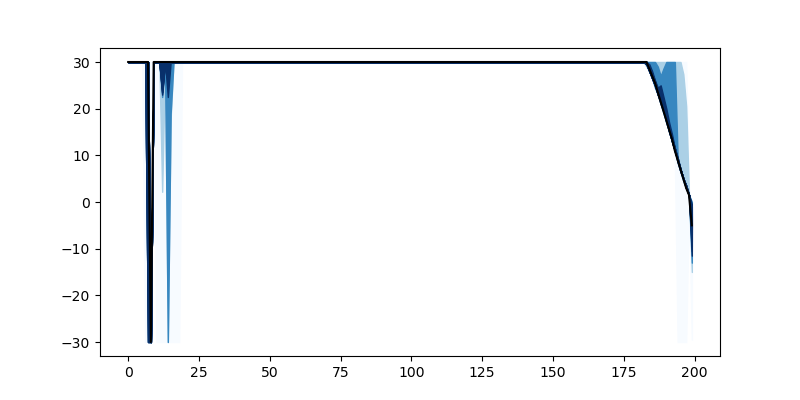}
  \caption*{\small Optimal control $Z$}
\end{minipage}
\begin{minipage}[c]{.29\linewidth}
  \includegraphics[width=0.99\linewidth]{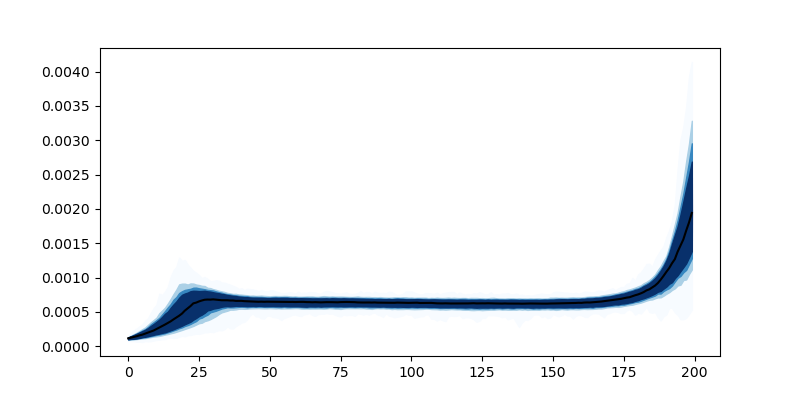}
  \caption*{\small Proportion $I$ of infected}
\end{minipage}
	\caption{ \small Dispersion of simulations of optimal government's controls, with the resulting trajectories of $I$.} \label{fig:SIRGen}
\end{figure}

\subsection{The first-best case}

First, remark that, with the particular choice of utility functions, we have
\begin{align*}
\chi^\star(\varpi)= \frac{1}{\theta_{\rm p}} \ln \bigg(\frac1\varpi -\phi_{\rm p} \bigg), \; \text{if} \; 0 < \varpi < \dfrac1{\phi_{\rm P}} = 2.
\end{align*}
Otherwise, if $\varpi \geq 2$, the optimal tax policy is equal to $- \infty$, which cannot be optimal from the government's point of view, since it leads to an infimum on $\varpi$ equal to $+ \infty$ (see \eqref{eq:V0_P_FB}).
For each value of the Lagrange parameter, a two dimensional PDE with a two-dimensional control $(\alpha,\beta)$ is considered.
A step discretisation for the grid in $(s,i)$ is taken equal to $(0.001,0.001)$. $A=[\eps,1]$ is discretised with 20 values and the values of $\beta$ are discretised with $80$ equally spaced values (to reduce the cost of optimisation).
We then search for the optimal $\varpi$ parameter with a step of $0.01$ within the interval $(0,2)$.
We obtain in this case an optimal value equal to $0.64$ and we give on \Cref{fig:firtsBest} the results, which show in particular that the epidemic is controlled in a similar way as in the second-best case, with incentives and testing policy.

 \begin{figure}[H]
\centering
    \begin{minipage}[c]{.29\linewidth}
          \includegraphics[width=0.99\linewidth]{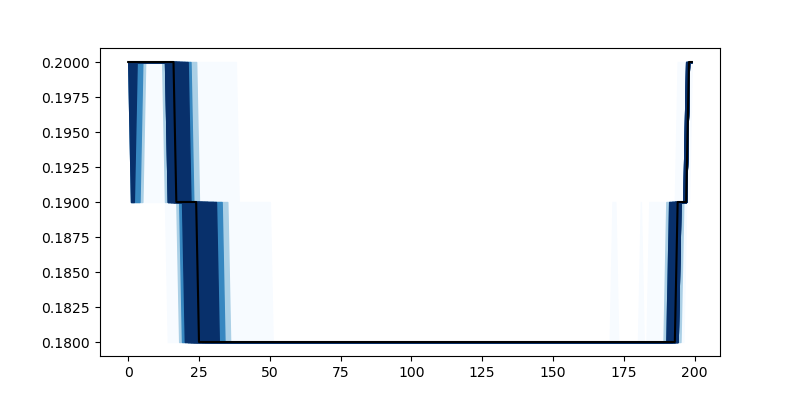}
          \caption*{\small Transmission rate $\beta$}
      \end{minipage}   
    \begin{minipage}[c]{.29\linewidth}
          \includegraphics[width=0.99\linewidth]{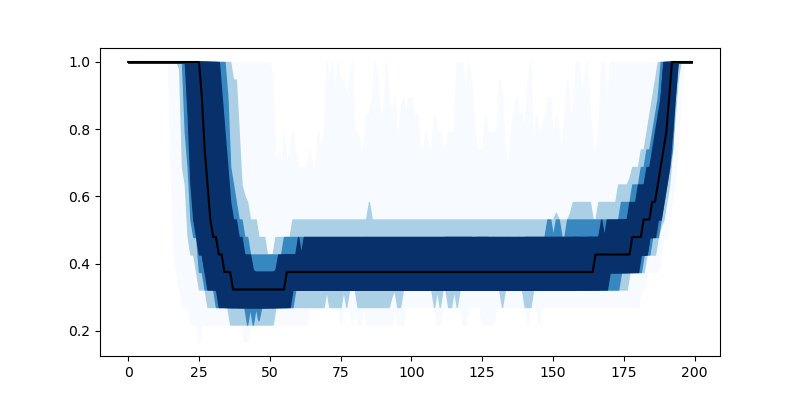}
          \caption*{\small Testing policy $\alpha$}
      \end{minipage}   
     \begin{minipage}[c]{.29\linewidth}
          \includegraphics[width=0.99\linewidth]{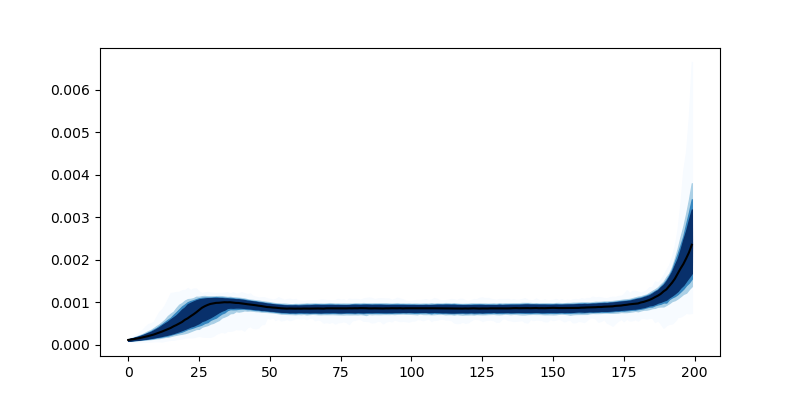}
          \caption*{\small Proportion of infected $I$}
      \end{minipage}  
      \caption{\small \label{fig:firtsBest} Dispersion of 500 trajectories obtained in the first-best case.}
 \end{figure}  
 
The shape of the optimal controls $\beta$ and $\alpha$, as well as the trajectories for the proportion $I$ of infected, are highly similar to those obtained in the previous case. The only clear difference is the principal's value. Indeed, we can compare the optimal value $V_0^{\rm P}$ for the government in the moral hazard case, to the first best value $V_0^{\rm P, FB}$. Using $10^4$ trajectories and the previously optimal control computed, we estimate $V_0^{\rm P,FB}= -0.249$ while $V_0^{\rm P}=-0.287$. The difference between the two values, with a relative difference of $15\%$ only pleads in favour of our incentive model: even without being able to track all the population, governments can achieve containment strategies with very similar levels of efficiency, and costs which are not significantly higher. This is of course partly explained by the fact that the testing is profitable both for the government and for the population, as it allows for values of $\beta$ close to $\overline \beta$, as shown on \Cref{fig:firtsBest}.

\section{Incentive policy for epidemic stochastic models}\label{sec:rigorous_maths}

\subsection{The stochastic model}\label{ss:weak_formulation}

\subsubsection{Initial canonical space}\label{sss:canonical_space_pop}

We fix a small parameter $\eps\in(0,1)$ to consider the subset $A:=[\eps,1]$. $\mathbb A$ is the set of all finite and positive Borel measures on $[0,T]\times A$, whose projection on $[0,T]$ is the Lebesgue measure. Every $q \in \mathbb A$ can be disintegrated as $q(\mathrm{d}s,\mathrm{d} v)=q_s(\mathrm{d} v)\mathrm{d}s$, for an appropriate Borel measurable kernel $(q_s)_{s\in[0,T]}$. We then define the following canonical space $\Omega:=\Cc^2\times\A,$ whose canonical process is denoted by $(S,I,\Lambda)$, in the sense that
\[
S_t\big(\mathfrak{s},\iota,q\big):=\mathfrak{s}(t),\; I_t\big(\mathfrak{s},\iota,q\big):=\iota(t),\; \Lambda\big(\mathfrak{s},\iota,q\big):=q,\; \forall\big(t,\mathfrak{s},\iota,q\big)\in[0,T]\times \Omega.
\]
We let $\Fc$ be the Borel $\sigma$-algebra on $\Omega$, and $\F:=(\Fc_t)_{t\in[0,T]}$ be the natural filtration of the canonical process
\begin{align*}
    \mathcal F_t:=\sigma\Big( \big(S_s, I_s, \Delta_s(\Upsilon) \big): (s,\Upsilon)\in[0,t]\times \Cc_b \big( [0,T]\times A, \R \big) \Big),\; t\in [0,T],
\end{align*}
where for any $(s,\Upsilon)\in[0,T]\times \Cc_b([0,T]\times A,\R)$, $\Delta_s(\Upsilon):=\iint_{[0,s]\times A} \Upsilon(r,a) \Lambda(\mathrm{d}r,\mathrm{d} a).$ Recall that in this framework $\Fc=\Fc_T$. Let $\M$ be the set of probability measures on $(\Omega,\Fc_T)$. For any $\P\in\M$, we let $\Nc^\P$ be the collection of all $\P$-null sets, that is to say $\Nc^\P:=\big\{N\in2^\Omega:\exists N^\prime\in\Fc_T,\; N\subset N^\prime,\; \P[N^\prime]=0\big\},$ where we recall that $2^\Omega$ represents the set of all subsets of $\Omega$, and we let $\F^\P:=(\Fc_t^\P)_{t\in[0,T]}$ be the $\P$-augmentation of $\F$, where $\Fc_t^\P:=\Fc_t\vee \sigma(\Nc^\P)$. We let $\F^{\P+}:=(\Fc_t^{\P+})_{t\in[0,T]}$ the corresponding right limit. Similarly, for any subset $\Pi\subset\M$, we let $\F^\Pi:=(\Fc^\Pi_t)_{t\in[0,T]}$ be the $\Pi$-universal completion of $\F$. Fix some initial values $(s_0,i_0) \in \R_+^2$.\footnote{Notice that the initial value of $r_0$ of $R$, which appears in the SIR version of the model, is irrelevant at this stage.}, and let us introduce the drift and volatility functions for our controlled model, namely $B: \R^2 \longrightarrow\R^2$ and $\Sigma: \R^2 \times A\longrightarrow \R^{2}$, defined by
\renewcommand{\arraystretch}{1}
\[
B(x,y):=\begin{pmatrix}
\lambda-\mu x+\nu y\\
-(\mu+\nu+\gamma+\rho)y
\end{pmatrix},\; \Sigma(x,y,a):=\begin{pmatrix}
\sigma axy\\
-\sigma axy
\end{pmatrix},\; (x,y,a)\in\R^2\times A,
\]
where the parameters $(\lambda,\mu,\nu,\gamma,\sigma)\in[0,\infty)^4\times\R_+^\star$ are given. For any $(s,\varphi)\in[0,T]\times \Cc^2_b( \R^2 ,\R)$, we set
\begin{align*}
    M_s(\varphi):=&\ \varphi(S_s,I_s) - \iint_{[0,s]\times A}\bigg( B(S_r,I_r) \cdot \nabla \varphi(S_r,I_r) + \frac12  {\rm Tr} \big[ D^2 \varphi(S_r,I_r) \big(\Sigma \Sigma^\top\big) (S_r,I_r,a) \big] \bigg)\Lambda(\mathrm{d}r, \mathrm{d} a).
\end{align*}

\begin{definition}\label{def:Pc}
We define the subset $\Pc \subset \M$ as the one composed of all $\P\in\M$ such that
\begin{enumerate}[label=$(\roman*)$]
    \item $M(\varphi)$ is an $(\mathbb F,\P)$--local martingale on $[0,T]$ for all $\varphi \in \Cc^2_b(\R^2,\R);$
    
    \vspace{-0.6em}
    \item $\P \big[(S_0,I_0) = (s_0,i_0)\big]=1;$
    
        \vspace{-0.6em}
    \item with $\P$-probability $1$, the canonical process $\Lambda$ is of the form $\delta_{\phi_\cdot}(\mathrm{d} v)$ for some Borel function $\phi:[0,T]\longmapsto A$, where as usual, for any $a\in A$, $\delta_{a}$ is the Dirac mass at $a$.
\end{enumerate}
\end{definition}

We can follow \citeauthor*{bichteler1981stochastic} \cite{bichteler1981stochastic}, or \citeauthor*{neufeld2014measurability} \cite[Proposition 6.6]{neufeld2014measurability} to define a pathwise version of the density of the quadratic variation of $S$, denoted by $\widehat \sigma:[0,T]\times\Omega\longrightarrow \R$, by $\widehat\sigma^2_t(\omega):=\underset{n\to\infty}{\mathrm{limsup}}\; n\big(\langle S\rangle_t(\omega)-\langle S\rangle_{t-1/n}(\omega)\big),\; (t,\omega)\in[0,T]\times\Omega.$ L\'evy's characterisation of Brownian motion ensures that the process\footnote{More precisely, one should first use the result of \citeauthor*{stroock1997multidimensional} \cite[Theorem 4.5.2]{stroock1997multidimensional} to obtain that on an enlargement of $(\Omega,\Fc_T)$, there is for any $\P\in\Pc$, a Brownian motion $W^\P$, and an $\F$-predictable process, $A$-valued process $\alpha^\P$ such that
\[
 S_t=s_0+\int_0^t (\lambda -\mu S_s +\nu I_s) \mathrm{d}s+\int_0^t\sigma \alpha^\P_sS_sI_s \mathrm{d}W^\P_s,\; t\in[0,T],\; \P\text{\rm--a.s.}
\]
The result for $W$ is then immediate. Notice in addition that since $W$ is defined as a stochastic integral, it should also depend on explicitly on $\P$. We can however use \citeauthor*{nutz2012pathwise} \cite[Theorem 2.2]{nutz2012pathwise} to define $W$ universally, as an $\F^\Pc$-adapted and continuous process. This requires some set-theoretic assumptions which we implicitly consider here, see \citeauthor*{possamai2018stochastic} \cite[Footnote 7]{possamai2018stochastic} for details.}
\begin{equation}\label{browniandef}
W_t:=\int_0^T\widehat\sigma_s^{-1/2}\mathbf{1}_{\hat{\sigma}_s\neq 0}\mathrm{d}S_s,\; t\in[0,T],
\end{equation}
is an $(\F^\Pc,\P)$--Brownian motion for any $\P\in\Pc$. For any $\P\in \Pc$, we denote by $\mathcal A_o(\mathbb P)$ the set of $\F$-predictable and $A$-valued process $\alpha:=(\alpha_s)_{s\in[0,T]}$ such that, $\P$--a.s.
\begin{equation}\label{eq:model_SI}
\begin{cases}
\displaystyle S_t=s_0+\int_0^t\big(\lambda -\mu S_s +\nu I_s\big)\mathrm{d}s+\int_0^t\sigma \alpha_sS_sI_s \mathrm{d}W_s,\; t\in[0,T],\\[0.8em]
\displaystyle  I_t=i_0-\int_0^t(\mu+\gamma+\nu+\rho) I_s \mathrm{d}s -\int_0^t\sigma \alpha_sS_sI_s \mathrm{d}W_s,\; t\in[0,T].\\
\end{cases}
 \end{equation}
Once again, it is a classical result (see for instance \citeauthor*{stroock1997multidimensional} \cite[Theorem 4.5.2]{stroock1997multidimensional}, or \citeauthor*{elie2021mean} \cite[Lemma 2.3]{elie2021mean}) that $\mathcal A_o(\mathbb P)$ is not empty. We recall that the term $\lambda\geq 0$ denotes the birth rate, the parameter $\mu\geq 0 $ is the natural death rate in the population (susceptible and infected), $\gamma\geq 0$ is the death rate inside the infected population. The parameters $\nu$ and $\rho$ correspond to recovery rates, depending on whether we are considering a SIS or a SIR model, see the remark below.

\begin{remark}
It can be noted that our model, which results from a mixing of the \textnormal{SIS} and \textnormal{SIR} models, can be interpreted as an \textnormal{SIR} model with partial immunisation, in the sense that only a part of the population develops antibodies for the disease after being infected. Thus, a proportion $\rho$ of the infected moves to the class $\rm R$, and can no longer be infected. Conversely, the proportion of the infected who do not develop antibodies reverts to the class $\rm S$, and can therefore contract the disease again. This resulting model is similar to the one developed by \textnormal{\citeauthor*{zhang2018epidemic} \cite{zhang2018epidemic}} and called \textnormal{SISRS}. 
This type of model seems in fact well suited to model epidemics related to new viruses, such as the \textnormal{COVID-19}, when the immunity of infected persons has not yet been proved.
\end{remark}
Before pursuing, we need a bit more notations, and will consider the following sets $\Ac_o:=\bigcup_{\P\in\Pc}\Ac_o(\P),$
as well as, for any $\alpha\in\Ac_o$, $\Pc(\alpha):=\big\{\P\in\Pc:\alpha\in\Ac_o(\P)\big\}.$
We will require that the controls chosen by the government lead to only one weak solution to \Cref{eq:model_SI}, and are such that the processes $S$ and $I$ remain non-negative. We will therefore concentrate our attention to the set $\Ac$ of admissible controls defined by
\[
\Ac:=\big\{\alpha\in\Ac_o:\Pc(\alpha)\; \text{\rm is a singleton $\{\P^\alpha\}$, and $(S,I)$ is $\R_+^2$-valued, $\P^\alpha$--a.s.}\big\}.
\]
Notice that the set $\Ac$ is not empty since any constant $A$-valued process automatically belongs to $\Ac$, as a direct consequence of \citeauthor*{gray2011stochastic} \cite[Section $3$]{gray2011stochastic} or \citeauthor*{gao2019dynamics} \cite[Lemma 2.3]{gao2019dynamics}. Remark then that, for any $\alpha\in\Ac$, we have $\widehat\sigma_t=\sigma S_tI_t\alpha_t,\; \mathrm{d}\P^\alpha\otimes\mathrm{d}t$--a.e., and
\[
S_t+I_t=s_0+i_0+\int_0^t\big(\lambda -\mu(S_s+I_s)-(\gamma+\rho) I_s \big)\mathrm{d}s,\; t\in[0,T],\; \P^\alpha\text{\rm--a.s.}
\]
We thus deduce, using the positivity of $S$ and $I$, that
\begin{align}\label{eq:bounded}
    0 \leq S_t+I_t=\mathrm{e}^{-\mu t}\big(s_0+i_0\big)+\int_0^t\mathrm{e}^{-\mu (t-s)}\big(\lambda -(\gamma+\rho) I_s \big)\mathrm{d}s
    \leq 
    F(t,s_0, i_0),\;  t\in[0,T],\; \P^\alpha\text{\rm--a.s.},
\end{align}
where for all $(t, s, i) \in [0,T] \times \R_+^2$
\begin{align}\label{eq:boundary_time}
F(t,s, i) := \mathrm{e}^{-\mu t}\big(s + i \big)+\lambda\bigg(\frac{1-\mathrm{e}^{-\mu t}}{\mu}\mathbf{1}_{\{\mu>0\}}+t\mathbf{1}_{\{\mu=0\}}\bigg).
\end{align}
This result proves in particular that $S$ and $I$ are actually $\P^\alpha$--almost surely bounded, for any $\alpha\in\Ac$. Moreover, if $(s_0, i_0) \in (\R_+^\star)^2$, then for all $t \in [0,T]$, both $S_t$ and $I_t$ are (strictly) positive.
\begin{remark}\label{rem:RI}
Note that in the \textnormal{SIR} model, described by the system \eqref{eq:SIS_SIR} with $\nu =0$, we have, for all $t \in [0,T]$, $R_t =r_0\mathrm{e}^{-\mu t}+ \rho \int_0^t I_s \mathrm{e}^{-\mu(t-s)}\mathrm{d}s,$ so that $R_t$ depends only on the observation of $I_s$ for $s\leq t$.  
In addition to that $0\leq S_t+I_t+R_t\leq \mathrm{e}^{-\mu t}\big(s_0+i_0+r_0\big)+\int_0^t\mathrm{e}^{-\mu (t-s)}\big(\lambda -\gamma I_s \big)\mathrm{d}s
\leq 
F(t,s_0, i_0)+r_0\mathrm{e}^{-\mu t}.$
\end{remark}
 
\subsubsection{Impact of the interaction}\label{ss:interaction_population}

The basic model from \eqref{eq:model_SI} takes into account the testing policy put into place by the government, but ignores so far the interacting behaviour of the population. We model this through an additional control process chosen by the population. More precisely, we fix some constant $\beta^{\rm max}>0$ representing the maximum rate of interaction that can be considered, and we define $B:=[0,\beta^{\rm max}]$. Let $\mathcal B$ be the set of all $\F$-predictable and $B$-valued processes. Given a testing policy $\alpha\in\Ac$ implemented by the government, notice that the following stochastic exponential
\[
\bigg(\exp\bigg(-\int_0^t \dfrac{\beta_s}{\sigma\sqrt{\alpha_s}}\mathrm{d}W_s-\frac12\int_0^t\dfrac{\beta_s^2}{\sigma^2\alpha_s}\mathrm{d}s\bigg)\bigg)_{t\in[0,T]},
\]
is an $(\F,\P^\alpha)$-martingale, given that the process $\beta/(\sigma\sqrt{\alpha})$ takes values in $\big[0,\beta^{\rm max}/(\sigma\sqrt{\eps})\big]$, $\P^\alpha$--a.s. Therefore, for any $(\alpha,\beta)\in\Ac\times\Bc$, we can define a probability measure $\P^{\alpha,\beta}$ on $(\Omega,\Fc)$, equivalent to $\P^\alpha$.
Using Girsanov's theorem, $W^\beta_t:=W_t+\int_0^t\frac{\beta_s}{\sigma\sqrt{\alpha_s}}\mathrm{d}s,\; t\in[0,T],$ is an $(\F,\P^{\alpha,\beta})$--Brownian motion, and \begin{equation}\label{eq:SIS2}\begin{cases}
\displaystyle S_t=s_0+\int_0^t\big(\lambda -\mu S_s +\nu I_s-\beta_s \sqrt{\alpha_s}S_sI_s\big)\mathrm{d}s+\int_0^t\sigma \alpha_sS_sI_s \mathrm{d}W^\beta_s,\; t\in[0,T],\\[0.8em]
\displaystyle  I_t=i_0-\int_0^t\big((\mu+\nu+\gamma+\rho) I_s-\beta_s\sqrt{\alpha_s} S_sI_s\big) \mathrm{d}s -\int_0^t\sigma \alpha_sS_sI_s \mathrm{d}W^\beta_s,\; t\in[0,T].\\
\end{cases}
\end{equation}

\subsubsection{Optimisation problems}\label{ss:def_contract}

At time $0$, the government informs the population about its testing policy $\alpha\in\Ac$, as well as its fine policy $\chi$, which for now will be an $\Fc_T$-measurable and $\R$-valued random variable (a set we denote by $\mathfrak C$). The population solves the following optimal control problem
\begin{equation}\label{pb:agent}
V_0^{\rm A}(\alpha,\chi):=\sup_{\beta\in \Bc} J_0^{\rm A} (\alpha, \chi,\beta), \; \text{with} \; J_0^{\rm A} (\alpha, \chi,\beta) := \mathbb E^{\P^{\alpha,\beta}}\bigg[\int_0^T u(t,\beta_t, I_t)\mathrm{d}t +U(-\chi)\bigg].
\end{equation}
The interpretation of the functions $u$ and $U$ is detailed in \Cref{sss:pop_problem}, where the population's problem was informally introduced. For any $(\alpha,\chi)\in\Ac\times\mathfrak C$, we recall that we denoted by $\Bc^\star(\alpha,\chi)$ the set of optimal controls for $V_0^{\rm A}(\alpha,\chi)$:
\begin{align}
\label{eq:Bc_star}
\Bc^\star(\alpha,\chi):=\big\{\beta\in\Bc:V_0^{\rm A}(\alpha,\chi)=J_0^{\rm A} (\alpha, \chi,\beta) \big\}.
\end{align}
We require minimal integrability assumptions at this stage, and insist that there exists some $p>1$ such that
\begin{equation}
\label{eq:integchi}
\E^{\P^\alpha}\big[|U(-\chi)|^p\big]<\infty, \; \text{for any} \; \alpha\in \Ac.
\end{equation}
\begin{remark}\label{rem:integ}
Notice that since for any $\alpha\in\Ac$ the Radon--Nykod\'ym density $\mathrm{d}\P^{\alpha,\beta}/\mathrm{d}\P^{\alpha}$ has moments of any order under $\P^\alpha$ $($since any $\beta\in\Bc$ is bounded and any $\alpha\in\Ac$ is bounded and bounded away from $0)$, a simple application of H\"older's inequality ensures that \eqref{eq:integchi} implies that for any $p^\prime\in(1,p)$ and any $\beta\in \Bc$, $\E^{\P^{\alpha,\beta}}\big[\big|U(-\chi)\big|^{p^\prime}\big]<\infty.$
\end{remark}

Recall that the government can only implement policies $(\alpha,\chi)\in\Ac\times \mathfrak C$ such that $V_0^{\rm A}(\alpha,\chi)\geq \underline v$, where the minimal utility $\underline v\in\R$ is given. We denote the subset of $\Ac\times \mathfrak C$ satisfying this constraint and \Cref{eq:integchi} by $\Xi$.

\medskip

In line with the informal reasoning developed in \Cref{sss:gov_problem}, the government aims at minimising the number of infected people until the end of the lockdown period, and we write rigorously its minimisation problem as
\begin{align}\label{pb:principal}
    V^{\rm P}_0 := \sup_{(\alpha,\chi)\in\Xi} \sup_{\beta\in \Bc^\star(\alpha,\chi)}
    \mathbb E^{\P^{\alpha,\beta}} \bigg[\chi - \int_0^T \big( c(I_t) + k(t, \alpha_t, S_t, I_t) \big)\mathrm{d} t \bigg],
\end{align}
where the functions $c : \R_+ \longrightarrow \R_+$ and $k : [0,T] \times A \times \R_+ \times \R_+ \longrightarrow \R$ were introduced in \Cref{sss:gov_problem}.

\subsection{Optimal interaction of the population given tax and test policies}\label{ss:solving_pop_pb}

\subsubsection{A relevant contract form}

Since the fine policy $\chi$ is an $\Fc_T$-measurable random variable, where $\F$ is the filtration generated by the process $(S,I)$, we should expect that in general $V_0^{\rm A}(\alpha,\chi)=v(0,s_0,i_0)$, where the map $v:[0,T]\times\Cc^2\longrightarrow \R$ satisfies an informal Hamilton Jacobi Bellman (HJB for short) equation, and as such has the dynamic
\[
\mathrm{d}v(t,S_t,I_t)=-H(S_t,I_t,Z_t^s,Z_t^i,\alpha_t)\mathrm{d}t+Z_t^s\mathrm{d}S_t+Z_t^i\mathrm{d}I_t,
\]
where the population's Hamiltonian $H:[0,T] \times (\R_+^\star)^2 \times\R^2\times A \longrightarrow \R$ is defined by
\begin{align*}
H(t, s,i,z,z^\prime,a) &:=\sup_{b\in B}h(t, s,i,z,z^\prime,a,b),\; (t, s, i, z, z^\prime,a) \in [0,T] \times (\R_+^\star)^2 \times\R^2\times A\\
\text{where } \; h(t, s,i,z,z^\prime,a,b) &:= \big(\lambda-\mu s+\nu i-b\sqrt{a}si\big) z - \big((\mu+\nu+\gamma+\rho)i - b\sqrt{a}si\big) z^\prime + u(t, b, i), \; \text{for} \; b \in B.
\end{align*}

In particular, defining $Z:=Z^s-Z^i$, we should have
\begin{align}\label{eq:utility_relevant_form}
U(-\chi) &= V_0^{\rm A} (\alpha, \chi) 
- \int_0^T H(t,S_t,I_t,Z_t^s,Z_t^i,\alpha_t) \mathrm{d}t
+ \int_0^T Z_t^s \mathrm{d} S_t
+\int_0^T Z_t^i \mathrm{d}I_t \nonumber \\
&=  V_0^{\rm A}(\alpha,\chi)
- \int_0^T \Big( (\mu+\nu+\gamma+\rho) I_t Z_t + \sup_{b\in B} \big\{ u(t,b,I_t) -b \sqrt{\alpha_t} S_t I_t Z_t \big\} \Big) \mathrm{d}t
- \int_0^T Z_t\mathrm{d}I_t.
\end{align}
Given the supremum appearing above, the following assumption will be useful for us.
\begin{assumption}\label{assump:argmax}
There exists a unique Borel-measurable map $b^\star: [0,T] \times \R_+^\star \times \R_+^\star\times \R\times A \longrightarrow B$ such that
\begin{align}\label{eq:b_star}
b^\star(t, s,i,z,a)\in\underset{b\in B}{\mathrm{argmax}} \big\{u(t,b,i) - b \sqrt{a} s i z \big\},\; \forall(t,s,i,z,a)\in [0,T] \times (\R_+^\star)^2 \times \R\times A.
\end{align}
\end{assumption}

\begin{remark}
We would like to insist on the fact that for the \textnormal{SIR} model and in view of \textnormal{\Cref{rem:RI}}, it is not necessary to consider that the process $R$ is a state variable. Indeed, its value at time $t$ can be deduced from the paths of $I$ until time $t$. More precisely, following the previous reasoning to find the relevant form of contracts, one could consider  
\[
\mathrm{d}v(t,S_t,I_t)=-\widetilde H(S_t,I_t,R_t,Z_t^s,Z_t^i,\alpha_t)\mathrm{d}t+Z_t^s\mathrm{d}S_t+Z_t^i\mathrm{d}I_t+Z_t^r \mathrm{d}R_t,
\]
where, in this case, the population's Hamiltonian $\widetilde H:[0,T] \times (\R_+^\star)^2 \times\R^2\times A$ is defined by
\begin{align*}
\widetilde H(t, s,i,r,z,z^\prime,\widetilde z,a) &:=\sup_{b\in B}\big\{ h(t, s,i,z,z^\prime,a,b)\big\} + (\rho i-\mu r) \widetilde z, \text{ for any } (t, s, i, z, z^\prime, \widetilde z,a) \in [0,T] \times (\R_+^\star)^2 \times \R^3\times A.
\end{align*}
Since the dynamics of $R$ is uncontrolled, a simplification occurs between the part of the Hamiltonian $(\rho i-\mu r) \widetilde z$ and the integral w.r.t. $\drm R$, which leads to the same form for the utility function as mentioned in \textnormal{\Cref{eq:utility_relevant_form}}.
\end{remark}

\subsubsection{The general analysis}
For any $(\alpha,m)\in\Ac\times\N^\star$, we define $\Sc^m(\P^\alpha)$ and $\H^m(\P^\alpha)$ as respectively the sets of $\R$-valued, $\F^{\P^\alpha+}$-adapted continuous processes $Y$ s.t. $\|Y\|_{\Sc^m(\P^\alpha)}<\infty$, and the set of $\F^{\P^\alpha}$-predictable, $\R$-valued processes $Z$ with $\|Z\|_{\H^m(\P^\alpha)}<\infty$, where
\[
\|Y\|_{\Sc^m(\P^\alpha)}^m:=\E^{\P^\alpha}\bigg[\sup_{t\in[0,T]}|Y_t|^m\bigg],\; \|Z\|^m_{\H^m(\P^\alpha)}:=\E^{\P^\alpha}\bigg[\bigg(\int_0^T\big|\widehat\sigma_sZ_s\big|^2\mathrm{d}s\bigg)^{m/2}\bigg],\; (Y,Z)\in\Sc^m(\P^\alpha)\times\H^m(\P^\alpha).
\]

\begin{theorem}\label{thm:agent}
Let $(\alpha, \chi) \in \Xi$. There exists a unique $\Fc^{\P^\alpha+}_0$-measurable $Y_0$ and a unique $Z \in \mathbb H^p (\P^\alpha)$ such that 
\begin{align}\label{eq:xi_representation}
    U(-\chi) = Y_0
    - \int_0^T \Big(Z_t (\mu+\nu+\gamma+\rho) I_t +  u(t, \beta^\star_t, I_t) - \beta^\star_t\sqrt{\alpha_t} S_t I_t Z_t \Big) \drm t
    - \int_0^T Z_t \drm I_t,\; \P^\alpha\text{\rm--a.s.},
\end{align}
with $\beta^\star_t:= b^\star(t, S_t, I_t, Z_t,\alpha_t)$ for all $t \in [0,T]$. Moreover, $\Bc^\star(\alpha,\chi)=\{\beta^\star\}$ and $V_0^{\rm A}(\alpha,\chi)=\E^{\P^{\alpha}}[Y_0]$.
\end{theorem}

\begin{proof}
Fix $(\alpha, \chi) \in \Xi$ as in the statement of the theorem. 
Let us consider the solution $(Y,Z)$ of the following BSDE
\begin{align}\label{BSDE:opt}
    Y_t =  U(-\chi) 
    + \int_t^T \sup_{b\in B}\big\{ u(r, b, I_r) - Z_rb \sqrt{\alpha_r} S_r I_r \big\} \drm r
    - \int_t^T Z_r \sigma \alpha_r S_r I_r \drm W_r, \; t \in [0,T].
\end{align}
Since $\chi \in \mathfrak C$, $u$ is continuous, $I$ and $S$ are bounded, and $B$ is a compact set, it is immediate this BSDE is well-posed and admits a unique solution $(Y, Z) \in \Sc^p(\mathbb P^\alpha) \times \H^p(\mathbb P^\alpha)$ (in a more general context, one may refer for instance to \citeauthor*{bouchard2018unified} \cite[Theorem 4.1]{bouchard2018unified}). Therefore, using the dynamic of $I$ under $\P^\alpha$, given by \Cref{eq:model_SI}, as well as the definition of $\beta^\star$, and letting $t=0$, we obtain that \eqref{eq:xi_representation} is satisfied. Next, using this representation for $U(\chi)$ in the population's criteria defined in \Cref{pb:agent}, we notice that, for any $\beta\in\Bc$,
\begin{align*}
J_0^{\rm A} (\alpha, \chi,\beta) &= \E^{\P^{\alpha,\beta}}\bigg[ Y_0 
+ \int_0^T \Big( u(t, \beta_t, I_t) - Z_t (\mu+\nu+\gamma+\rho) I_t -  u(t, \beta^\star_t, I_t) + \beta^\star_t \sqrt{\alpha_t} S_t I_t Z_t \Big) \drm t
- \int_0^T Z_t \drm I_t\bigg]\\
&= \E^{\P^\alpha}[Y_0] + \sup_{\beta\in\Bc} \E^{\P^{\alpha,\beta}}\bigg[  \int_0^T \Big( u(t, \beta_t, I_t) - \beta_t S_t I_t Z_t - u(t, \beta^\star_t, I_t) + \beta^\star_t\sqrt{\alpha_t} S_t I_t Z_t \Big) \drm t \bigg]\leq \E^{\P^\alpha}[Y_0],
\end{align*}
where we used the fact that that $Z \in \mathbb H^p (\P^\alpha)$, and that $\Ec^\beta_\cdot:=\exp\big(-\int_0^\cdot\frac{\beta_s}{\sigma\sqrt{\alpha_s}}\mathrm{d}W_s-\frac12\int_0^\cdot\frac{\beta_s^2}{\sigma^2\alpha_s}\mathrm{d}s\big),$ is continuous, and both an $(\F^{\P^\alpha},\P^\alpha)$- and an $(\F^{\P^\alpha+},\P^\alpha)$-martingale (see \citeauthor*{neufeld2014measurability} \cite[Proposition 2.2]{neufeld2014measurability}), so that for any $\beta\in\Bc$
\begin{align*}
\E^{\P^{\alpha,\beta}}[Y_0] &= 
\E^{\P^{\alpha}}\big[\Ec_T^\beta Y_0\big]
=\E^{\P^{\alpha}}\big[\Ec_0^\beta Y_0\big]
=\E^{\P^{\alpha}}[Y_0].
\end{align*}
The previous inequality implies that $
V_0^{\rm A}(\alpha,\chi)\leq \E^{\P^\alpha}[Y_0].$
Moreover, thanks to \Cref{assump:argmax}, equality is achieved if and only if we choose the control $\beta^\star$. This shows that $
V_0^{\rm A}(\alpha,\chi)=\E^{\P^{\alpha}}[Y_0], \; \text{\rm and}\; \Bc^\star(\alpha,\chi)=\big\{\beta^\star\big\}.$
\end{proof}
In the previous result, the fact that \Cref{eq:xi_representation} holds with an $\Fc_0^{\P^\alpha+}$-measurable random variable and not a constant is somewhat annoying. The next lemma shows that we can actually have the representation with a constant without loss of generality.
\begin{lemma}\label{lemma:rep}
Let $\alpha\in\Ac$, and fix an $\Fc^{\P^\alpha+}_0$-measurable random variable $Y_0$ and some $Z \in \mathbb H^p (\P^\alpha)$. Define the following contracts
\begin{align*}
    \chi &:=-U^{(-1)}\bigg(Y_0
    - \int_0^T \Big(Z_t (\mu+\nu+\gamma+\rho) I_t +  u(t, \beta^\star_t, I_t) - \beta^\star_t \sqrt{\alpha_t} S_t I_t Z_t \Big) \drm t
    - \int_0^T Z_t \drm I_t\bigg),\\
     \chi^\prime &:= -U^{(-1)}\bigg( \E^{\P^\alpha}[Y_0]
    - \int_0^T \Big(Z_t (\mu+\nu+\gamma+\rho) I_t +  u(t, \beta^\star_t, I_t) - \beta^\star_t \sqrt{\alpha_t} S_t I_t Z_t \Big) \drm t
    - \int_0^T Z_t \drm I_t\bigg).
\end{align*}
Then $V_0^{\rm A}(\alpha,\chi)=V_0^{\rm A}(\alpha,\chi^\prime)=\E^{\P^\alpha}[Y_0],\; \Bc^\star(\alpha,\chi)=\Bc^\star(\alpha,\chi^\prime)=\big\{\beta^\star\big\}.$
\end{lemma}
\begin{proof}
The equalities for $(\alpha,\chi)$ are immediate from \Cref{thm:agent}. For $(\alpha,\chi^\prime)$, we have, using the fact that $Z\in\H^p(\P^\alpha)$, and thus $Z\in\H^q(\P^{\alpha,\beta})$ for any $\beta\in\Bc$ and any $q\in(1,p)$
\begin{align*}
V_0^{\rm A}(\alpha,\chi^\prime) &= \sup_{\beta\in\Bc}\E^{\P^{\alpha,\beta}}\bigg[ \E^{\P^\alpha}[Y_0]
    + \int_0^T \Big(u(t, \beta_t, I_t) - Z_t (\mu+\nu+\gamma+\rho) I_t - u(t, \beta^\star_t, I_t) + \beta^\star_t \sqrt{\alpha_t} S_t I_t Z_t \Big) \drm t - \int_0^T Z_t \drm I_t \bigg]\\
    &=\E^{\P^\alpha}[Y_0]+\sup_{\beta\in\Bc}\E^{\P^{\alpha,\beta}}\bigg[  \int_0^T \Big(u(t, \beta_t, I_t) - \beta_t \sqrt{\alpha_t} S_t I_t Z_t - u(t, \beta^\star_t, I_t) + \beta^\star_t\sqrt{\alpha_t} S_t I_t Z_t \Big) \drm t\bigg]\leq \E^{\P^\alpha}[Y_0].
\end{align*}
Since the equality is attained if and only if we choose $\beta=\beta^\star$, this ends the proof.
\end{proof}

\subsubsection{Characterisation of the class of admissible contracts} 

We introduce the class $\overline \Xi$ of contracts defined by all pairs $(\alpha,\chi^{y_0,Z})$ with $\alpha\in\mathcal A $ and $\chi^{y_0,Z} := -U^{(-1)}(Y_T^{y_0,Z})$, where $Y^{ y_0,Z}$ is a process given, $\P^\alpha\text{\rm--a.s.}$, for all $t\in[0,T]$ by 
\[
Y_t^{y_0,Z} = y_0 - \int_0^t \Big(Z_r (\mu+\nu+\gamma+\rho) I_r +  u\big(t,b^\star(r, S_r,I_r,Z_r,\alpha_r), I_r\big) - b^\star(r, S_r,I_r,Z_r,\alpha_r)\sqrt{\alpha_r} S_r I_r Z_r \Big) \drm r
- \int_0^t Z_r \drm I_r,
\]
with $Z\in\mathbb H^p(\mathbb P^\alpha)$ and $y_0\in[\underline v,\infty)$. We also denote for simplicity $\P^{\star,\alpha, Z}:=\P^{\alpha,b^\star(S_\cdot,I_\cdot,Z_\cdot)}$.
    
\begin{lemma}\label{lem:control_pb_principal}
The problem of the government given by \eqref{pb:principal} can be rewritten
\begin{equation}\label{pb:principal:reduce}
V^{\rm P}_0 = \sup_{(\alpha,Z)\in\mathcal A\times \mathbb H^p(\mathbb P^\alpha)} \mathbb E^{\P^{\star,\alpha,Z}}
\bigg[ -U^{(-1)} \big( Y_T^{\underline v,Z} \big)
- \int_0^T \big( c(I_s) + k(s,\alpha_s, S_s, I_s) \big)\mathrm{d}s\bigg].
\end{equation}
\end{lemma}

\begin{proof}
From \Cref{thm:agent} and \Cref{lemma:rep}, we know that $\Xi\subset \overline \Xi$. To prove the reverse inclusion, let us now consider a pair $(\alpha, \chi^{y_0,Z})\in \overline \Xi$. In fact, to show that $\overline\Xi \subset \Xi$ (and thus that $\Xi=\overline\Xi$), we simply need to ensure that $\chi^{y_0,Z}$ satisfies the integrability condition \eqref{eq:integchi}. Using the fact that $u$ is continuous, $B$ is compact, $\alpha$ is bounded below by $\eps$, and $S$ and $I$ are bounded, we have that there exists a constant $C>0$, which may change value from line to line, such that
\begin{align*}
\mathbb E^{\mathbb P^\alpha}\big[\big|U( -\chi^{y_0,Z}) \big|^p\big]
&\leq C\bigg(1+ \mathbb E^{\mathbb P^\alpha}\bigg[\bigg(\int_0^T|S_rI_rZ_r|\mathrm{d}r\bigg)^p+ \bigg|\int_0^T \widehat\sigma_rZ_r\mathrm{d}W_r\bigg|^p\bigg]\bigg)\\
& \leq C\bigg(1+\mathbb E^{\mathbb P^\alpha}\bigg[\bigg(\int_0^T\sigma\alpha_r|S_rI_rZ_r|\mathrm{d}r\bigg)^p \bigg]+ \|Z\|^p_{\H^p(\P^\alpha)}\bigg)\leq C\big(1+ \|Z\|^p_{\H^p(\P^\alpha)}\big)<\infty,
\end{align*}
where we used Burkholder--Davis--Gundy's inequality and Cauchy--Schwarz's inequality, implying that \eqref{eq:integchi} holds.

\medskip
Next, we use \Cref{lemma:rep} to realise that $\Bc^\star(\alpha, \chi^{y_0,Z})=\big\{b^\star(\cdot, S_\cdot,I_\cdot,Z_\cdot,\alpha_\cdot)\big\}$, and $V_0^{\rm A}(\alpha, \chi^{y_0,Z})=y_0$, which implies
\[
 V^{\rm P}_0=\sup_{y_0\geq \underline v}\sup_{(\alpha,Z)\in\mathcal A\times \mathbb H^p(\mathbb P^\alpha)} 
 \mathbb E^{\P^{\star,\alpha,Z}} \bigg[ -U^{(-1)} \big( Y_T^{ y_0,Z} \big) 
 - \int_0^T \big( c(I_s) + k(s, \alpha_s, S_s, I_s) \big) \mathrm{d}s \bigg].
\]
To conclude, it is enough to notice that the following map is non-increasing
\begin{align*}
[\underline v,\infty)\ni y_0 \longmapsto  \mathbb E^{\P^{\star,\alpha,Z}}\bigg[-U^{(-1)}\big( Y_T^{y_0,Z}\big)- \int_0^T \big( c(I_s) + k(s, \alpha_s, S_s, I_s) \big) \mathrm{d}s \bigg] \in\R.
\end{align*}\end{proof}

\subsection{Optimal tax and test policies under moral hazard for epidemic models}\label{ss:solving_gov_pb}

\subsubsection{Weak formulation for the government's problem}

\Cref{lem:control_pb_principal} states that the problem of the government can be can be reduced to a more standard stochastic control problem. However, in the current formulation, one of the three state variables, namely $Y$, is considered in the strong formulation, while the other state variables $S$ and $I$ are considered in weak formulation. Indeed, the variable $Y$ is indexed by the control $Z$, while the control $(\alpha, Z)$ only impacts the distribution of $S$ and $I$ through $\P^{\star, \alpha, Z}$. As highlighted by \citet*[Remark 5.1.3]{cvitanic2012contract}, it makes little sense to consider a control problem of this form directly. Therefore, contrary to what is usually done in principal--agent problems (see, \textit{e.g.}, \cite{cvitanic2018dynamic}), we decided to adopt the weak formulation to rigorously write the problem of the principal, since this is the formulation which makes sense for the agent's problem. We will thus formulate it below, for the sake of thoroughness.\footnote{Notice that at the end of the day, this is not really an issue. Indeed, provided that the problem has enough regularity (typically some semi-continuity of the terminal and running reward with respect to state), one can expect the strong and weak formulations to coincide. See for instance \citeauthor*{karoui2013capacities2} \cite[Theorem 4.5]{karoui2013capacities2}.}

\medskip

Let $V := \R \times A$ and consider the sets $\V$ as we defined $\A$ in \Cref{sss:canonical_space_pop}. The intuition is that the principal's problem depends only on time and on the state variable $X = (S,I,Y)$. Following the same methodology used for the agent's problem, to properly define the weak formulation of the principal's problem, we are led to consider the canonical space $\Omega^{\rm P} :=  \Cc^3 \times \V$, with canonical process $(S, I, Y, \Lambda^{\rm P})$, where for any $(t, \mathfrak{s}, \iota, y, q) \in[0,T] \times \Omega^{\rm P}$:
\[
S_t (\mathfrak{s},\iota,y,q):=\mathfrak{s}(t),\; 
I_t (\mathfrak{s},\iota,y,q):=\iota(t),\;
Y_t (\mathfrak{s},\iota,y,q):=y(t), \;
\Lambda^{\rm P} (\mathfrak{s},\iota,y,q):=q.
\]
We let $\Gc$ be the Borel $\sigma$-algebra on $\Omega^{\rm P}$, and $\G:=(\Gc_T)_{t\in[0,T]}$ the natural filtration of $(S, I, Y, \Lambda^{\rm P})$, defined in the same way as $\F$ in the previous canonical space $\Omega$ (see \Cref{ss:weak_formulation}). Let then $\M^{\rm P}$ be the set of probability measures on $(\Omega^{\rm P},\Gc_T)$. For any $\P \in \M^{\rm P}$, we can define $\G^\P$ the $\P$-augmentation of $\G$, its right limit $\G^{\P+}$, as well as $\F^\Pi:=(\Fc_t^{\Pi})_{t\in[0,T]}$ the $\Pi$-universal completion of $\F$ for any subset $\Pi\subset\M^{\rm P}$.

\medskip

The drift and volatility functions for the process $X$ are now defined for any $(t, s, i, z,a) \in [0,T] \times (\R_+^\star)^2 \times V$
\renewcommand{\arraystretch}{1}
\begin{align}
\label{eq:edsGen}
    B^{\rm P} (t, s, i, z,a):=
    \begin{pmatrix}
        \lambda - \mu s +\nu i - b^\star(t, s,i,z,a)\sqrt{a} s i \\[0.3em]
        - (\mu+\nu+\gamma+\rho) i + b^\star(t, s,i,z,a)\sqrt{a} s i \\[0.3em]
        - u^\star (t, s,i,z,a)
    \end{pmatrix},\; 
    \Sigma^{\rm P} (s, i, z,a):= \sigma a s i
    \begin{pmatrix}
        1 \\[0.3em]
        -1 \\[0.3em]
        z
    \end{pmatrix},
\end{align}
where $u^\star (t, s,i,z,a) := u (t, b^\star(t, s,i,z,a), i)$, for all $(t, s, i, z) \in [0,T] \times (\R_+^\star)^2 \times \R$.
For any $(t,\varphi^{\rm P}) \in [0,T] \times \Cc^2_b( \R^3,\R)$, we define
\begin{align*}
    M^{\rm P}_t(\varphi^{\rm P}):=&\ \varphi^{\rm P}(X_t) - \iint_{[0,t]\times V}\bigg( B^{\rm P} (r, S_r,I_r,  v) \cdot \nabla \varphi^{\rm P}(X_r) + \frac12  {\rm Tr} \big[ D^2 \varphi^{\rm P}(X_r) \big(\Sigma^{\rm P} ( \Sigma^{\rm P} )^\top \big) (r, S_r,I_r, v) \big] \bigg)\Lambda^{\rm P}(\mathrm{d}r, \mathrm{d} v).
\end{align*}

In the spirit of \Cref{def:Pc} for $\Pc \subset \M$, we define the subset $\Qc \subset \M^{\rm P}$ as the one consisting of all $\P\in\M^{\rm P}$ such that
\begin{enumerate}[label=$(\roman*)$]
    \item $M^{\rm P}(\varphi^{\rm P})$ is a $(\G,\P)$--local martingale on $[0,T]$ for all $\varphi^{\rm P} \in \Cc^2_b(\R^3,\R)$;
    \item $\P \big[X_0 = x_0 \big]=1$, where $x_0 := (s_0,i_0, \underline v)$;
    \item with $\P$-probability $1$, the canonical process $\Lambda^{\rm P}$ is of the form $\delta_{\phi_\cdot}(\mathrm{d} v)$ for some Borel-measurable function $\phi:[0,T]\longmapsto V$.
\end{enumerate}

Still following the line of \Cref{ss:weak_formulation}, we know that for any $\P \in \Qc$, we can define a $(\G^\Qc, \P)$--Brownian motion $W^{\rm P}$. We then denote by $\Vc_o(\P)$ the set of $\G$-predictable and $V$-valued process $(Z, \alpha)$ such that, $\P$--a.s. and for all $t\in[0,T]$,
\begin{equation}\label{eq:SIY}\begin{cases}
\displaystyle S_t = s_0 
+ \int_0^t \big(\lambda -\mu S_r +\nu I_r - b^\star (r, S_r, I_r, Z_r,\alpha_r)\sqrt{\alpha_r} S_r I_r \big) \drm r
+ \int_0^t \sigma \alpha_r S_r I_r \drm W^{\rm P}_r,\\[0.8em]
\displaystyle  I_t = i_0 
- \int_0^t \big((\mu+\nu+\gamma+\rho) I_r - b^\star(r, S_r, I_r, Z_r,\alpha_r)\sqrt{\alpha_r} S_r I_r \big) \drm r
- \int_0^t \sigma \alpha_r S_r I_r \drm W^{\rm P}_r,\\[0.8em]
\displaystyle  Y_t = \underline v 
- \int_0^t u^\star (r, S_r,I_r,Z_r,\alpha_r) \drm r
+ \int_0^t Z_r \sigma \alpha_r S_r I_r  \drm W^{\rm P}_r.
\end{cases}
\end{equation}

\subsubsection{Solving the government's problem}\label{sss:HJB_principal}

Thank to the analysis conducted in the previous subsection, the problem of the government given by \eqref{pb:principal} can now be written rigorously in weak formulation
\begin{align}\label{eq:pb_principal_weak}
    V^{\rm P}_0 = \sup_{\P \in \Qc} \mathbb E^{\P} \bigg[
    -U^{(-1)} ( Y_T) - \int_0^T \big(c(I_s) + k(s, \alpha_s, S_s, I_s) \big) \drm s\bigg].
\end{align}

We then define the Hamiltonian of the government, for all $t \in [0,T]$, $x := (s,i,y) \in \R^3$ and $(p,M) \in \R^3 \times \S^{3}$, by
\begin{align}\label{eq:hamiltonian_principal}
    H^{\rm P} (t, x, p, M) := \sup_{ (z,a) \in V} \bigg\{ B^{\rm P} (t, s, i, z,a) \cdot p + \dfrac{1}{2} {\rm Tr} \Big[ M \big( \Sigma^{\rm P} (\Sigma^{\rm P})^\top \big) (t, s, i, z,a) \Big] - k(t, a, s,i) \bigg\} - c(i),
\end{align}
where $\S^3$ represents the set of $3 \times 3$ symmetric positive matrices with real entries. More explicitly, the Hamiltonian can be written as follows, with $f(z,M) := M_{11} - 2 M_{12} + M_{22} - 2 z (M_{23} -M_{13}) + z^2 M_{33}$ for all $(z,M) \in \R \times \S^{3}$:
\begin{align*}
    H^{\rm P} (t, x, p, M)
    = &\ \sup_{z \in \R, a\in A} \bigg\{ b^\star(t, s,i,z,a) \sqrt{a} s i (p_2 - p_1)
    - u^\star (t, s,i,z,a) p_3
    + 
    \frac12\sigma^2 a^2 (si)^2 f(z,M) 
    -  k(t, a, s, i) \bigg\} \\[0.3em]
    &+ (\lambda - \mu s +\nu i) p_1 
    - (\mu+\nu+\gamma+\rho) i p_2
    - c (i).
\end{align*}
We are then led to consider the following HJB equation, for all $t \in [0,T]$ and $x = (s,i,y) \in \R^3$:
\begin{align}\label{eq:HJB_principal}
    - \partial_t v (t, x)  - H^{\rm P} (t, x, \nabla_x v, D^2_x v) = 0,\; (t,x)\in \Oc,
\end{align}
with terminal condition $v (T,x) := -U^{(-1)}(y)$, and where, recalling that $F$ is defined by \eqref{eq:boundary_time}, the natural domain over which the above PDE must be solved is\footnote{The boundary of the domain cannot be reached by the processes $S$ and $I$, which is why it not necessary to specify a boundary condition. Notice though that the upper bound can formally only be attained when $I$ is constantly $0$, in which case $S$ becomes deterministic, and the government best choice for $\alpha$ is clearly $1$, and its choice of $Z$ becomes irrelevant. In such a situation, we would immediately have $V_0^{\rm P}=\underline v$.}
 $\Oc:=\big\{(t,s,i,y)\in[0,T)\times\R_+^2\times \R: 0<s+i<F(t,s_0,i_0)\big\}.$

\begin{remark}
Standard arguments from viscosity solution theory allow to prove that $V_0^{\rm P} = v^{\rm P}(0, x_0)$ $($recalling that $x_0 = (s_0, i_0, \underline v))$ where $v^{\rm P}$ should be understood as the unique viscosity solution, in an appropriate class of functions, of the {\rm PDE} \eqref{eq:HJB_principal}. Obtaining further regularity results is by far more challenging. Indeed, it is a second-order, fully non-linear, parabolic {\rm PDE}, which is clearly not uniformly elliptic, the corresponding diffusion matrix being degenerate. This makes the question of proving the existence of an optimal contract a very complicated one, which is clearly outside the scope of our study. As a sanity check though, we recall that $\eps$-optimal contracts always exist, and can be indeed approximated numerically. See for instance {\rm \citeauthor*{kharroubi2020regulation} \cite{kharroubi2020regulation}} for an explicit construction of such $\eps$-optimal contracts in a particular case dealing with the stochastic logistic equation.
\end{remark}

\subsubsection{Comparison with the first-best case}\label{sss:fb_maths}

As already mentioned, the first-best case corresponds to the case where the government can enforce whichever interaction rate $\beta\in\Bc$ it desires (in addition to a contract $(\alpha,\chi) \in \Ac \times \mathfrak C$), and simply has to satisfy the participation constraint of the population. In order to find the optimal interaction rate in this scenario, as well as the optimal contract, one has to solve the government's problem defined by \eqref{eq:FB}.

\medskip

The simplest way to take into account the inequality constraint in the definition of $V_0^{\rm P,FB}$ is to introduce the associated Lagrangian. By strong duality, we then have
\[
V_0^{\rm P,FB}=\inf_{\varpi> 0}\sup_{(\alpha,\chi,\beta)\in\Ac\times\mathfrak C\times \Bc}\bigg\{\E^{\P^{\alpha,\beta}}\bigg[\chi-\int_0^T\big(c(I_t)+k(t,\alpha_t,S_t,I_t)\big)\mathrm{d}t\bigg]+\varpi\bigg(\mathbb E^{\P^{\alpha,\beta}} \bigg[\int_0^T u(t,\beta_t, I_t)\mathrm{d}t + U(-\chi)\bigg]- \underline v\bigg)\bigg\}.
\]
First, by concavity of $U$, it is immediate that for any given Lagrange multiplier $\varpi>0$, the optimal tax is constant. Then, using the definition of $\overline V_0(\varpi)$
for any $\varpi>0$ in \eqref{eq:V0_varpi_FB},
we have:
\[
V_0^{\rm P,FB}=\inf_{\varpi>0}\Big\{\chi^\star(\varpi)+\varpi \big(U\big(-\chi^\star(\varpi)\big)-\underline v\big)
+ \overline V_0(\varpi) \Big\}.
\]
Note that $\overline V_0(\varpi)$ is the value function of a standard stochastic control problem. Therefore,
we expect to have $\overline V_0(\varpi)= v^\varpi(0,s_0,i_0)$, where the function $v^\varpi:[0,T]\times\R_+^2 \longrightarrow \R$ solves the following HJB PDE
\[
\begin{cases}
\displaystyle
-\partial_t v^\varpi(t,s,i) + c(i) 
- (\lambda-\mu s+\nu i) \partial_s v^\varpi
+ (\mu+\nu+\gamma+\rho)i \partial_i v^\varpi
- \Hc^{\varpi} (t,s,i,\partial v^\varpi, D^2 v^\varpi) =0,\; (t,s,i)\in \Dc,\\[0.5em]
v^\varpi(T,s,i)=0,\; (s,i)\in\Dc_T,
\end{cases}
\]
where the Hamiltonian is defined, for $t \in [0,T]$, $(s,i) \in (\R_+^\star)^2$, $p :=(p_1,p_2) \in \R^2$ and $M \in \S^2$ by
\begin{align*}
    \Hc^{\varpi} (t,s,i,p,M) := \sup_{a \in A} \bigg\{
\sup_{b \in B} \big\{
\varpi u(t,b,i) - b si \sqrt{a}  (p_1-p_2)
\big\} - k(t,a,s,i) 
+ \dfrac12 \sigma^2 (si)^2 a^2 (M_{11} -2 M_{12} + M_{22})
\bigg\}.
\end{align*}

Note that if we consider separable utilities with the forms in \Cref{ex:utility_pop}, the optimal interaction rate is given, for a testing policy $\alpha \in \Ac$ and a Lagrange multiplier $\varpi > 0$, by $\beta^{\varpi}_t = b^{\varpi} \big(S_t, I_t, \partial v^\varpi (t,S_t,I_t), \alpha_t \big)$ for all $t \in [0,T]$, where
\begin{align*}
    b^{\varpi} (s,i,p,a) := b^\circ \big( s, i,\sqrt{a} (p_1-p_2) / \varpi \big), \; \text{ for all } \; 
    (s,i,p,a) \in (\R_+^\star)^2 \times \R^2 \times A.
\end{align*}

\section{Extensions and generalisations}\label{sec:extensions}

\subsection{Diseases with latency periods: SEIS, SEIR}

The reasoning developed in this paper can be extended in a straightforward way to consider SEIR and SEIS compartment models. These models are used to describe epidemics in which individuals are not directly contagious after contracting the disease, as for the COVID-19 epidemic (see, \textit{e.g.}, \citeauthor{dolbeault2020heterogeneous} \cite{dolbeault2020heterogeneous}), and thus involve a fourth class representing the `Exposed', \textit{i.e.}, individuals who have contracted the disease but are not yet infectious. The constant rate at which an exposed person becomes infectious is denoted by $\iota \in \R_+$.
The difference between SEIS and SEIR models is embedded into the immunity toward the disease: for SEIR, it is assumed that the immunity is permanent (as in a SIR), whereas for SEIS, infected individual come back in the susceptible class at rate $\nu\geq 0$, similarly to SIS models. We can also take into account the demographic dynamics of the population, through the parameters $\lambda$, $\mu$ and $\gamma$. 
Similarly to the previous models, we consider that the dynamic of the epidemic is subject to a noise in the estimation of the proportion of susceptible and infected individuals. 
Inspired by the stochastic model in \citeauthor*{mummert2019parameter} \cite[Equation (3)]{mummert2019parameter}, we can consider that the dynamics of the epidemic is given by:
\begin{equation}\label{eq:SEIRS_sto}
\begin{cases}
\displaystyle S_t = s_0 +\int_0^t\big(\lambda - \mu S_s - \beta_s \sqrt{\alpha_s}S_s I_s + \nu I_s \big) \drm s + \int_0^t \sigma \alpha_sS_sI_s \drm W_s,\\[0.8em]
\displaystyle E_t =e_0-\int_0^t  \big( (\mu+\iota) E_s - \beta_s \sqrt{\alpha_s}S_s I_s \big) \mathrm{d}s -\int_0^t \sigma \alpha_sS_sI_s \drm W_s,\\[0.8em]
\displaystyle I_t = i_0- \int_0^t \big( (\mu + \nu +\gamma+\rho) I_s -\iota E_s \big) \mathrm{d}s,\\[0.8em]
\displaystyle R_t =r_0+\int_0^t (\rho I_s- \mu R_s) \mathrm{d}s,
\end{cases}
\text{for } \; t\in[0,T],
\end{equation}
Note that the proportion $I$ of infected and infectious is uncertain, but only through its dependence on $E$ and the proportion $R$ of recovery is uncertain only through its dependence on $I$. More precisely, we assume that there is no uncertainty on both the recovery rate $\rho$, the rate $\iota$ at which infected people becomes infectious and the (potentially) rate $\nu$ at which an individual loses immunity, implying that if the proportion of exposed individual is perfectly known, the proportion of infected is also known without uncertainty and consequently the proportion of recovery is also certainly known. Again this modelling choice is consistent with most stochastic SEIRS models, and emphasises that the major uncertainty in the current epidemic is related to the non-negligible proportion of (nearly) asymptomatic individuals. Indeed, an asymptomatic individual may be misclassified as susceptible or exposed.

\medskip

We will now give, informally, the optimisation problems faced by both the population and the government. 
The most important change compared to SIS/SIR models is that the criteria should now depend on the sum $E+I$, representing the proportion of the population having contracted the disease, rather than just the proportion $I$ of infectious people. For example, we can consider the following form for the population's problem:
\begin{align*}
V_0^{\rm A} (\alpha,\chi) := \sup_{\beta\in \Bc} \mathbb E \bigg[\int_0^T u\big(t,\beta_t, E_t+I_t\big)\mathrm{d}t +U(-\chi)\bigg],
\end{align*}
while that of the government could become
\begin{align*}
    V^{\rm P}_0 := \sup_{(\alpha,\chi) \in \Xi} \sup_{\beta \in \Bc^\star (\alpha,\chi)}
    \mathbb E \bigg[\chi - \int_0^T \big( c\big(E_t+I_t\big) + k(t, \alpha_t, S_t, I_t) \big)\mathrm{d} t \bigg].
\end{align*}
A slight adaption of our earlier arguments will show that admissible taxes take the form $\chi := - U^{(-1)} (Y_T)$ with
\begin{align*}
    Y_t := Y_0
    - \int_0^T \Big(Z_t (\mu+\iota) E_t +  u(t, \beta^\star_t, E_t+I_t) - \beta^\star_t\sqrt{\alpha_t} S_t I_t Z_t \Big) \drm t
    - \int_0^T Z_t \drm E_t,
\end{align*}
where $\beta^\star$ is the population's optimal contact rate, under the assumption it exists. It thus remain to solve the government's problem, but unlike in the previous SIS/SIR models, there are now four state variables, namely $(S, E, I,Y)$. 
However, solving it numerically is really more challenging since it increases the dimension of the problem. A numerical investigation seems to be complicated as far as we now, and we left these numerical issues for future research. 

\subsection{Beyond SEIS/SEIR models: a theoretically tractable method}

There are of course plethora of generalisations of the models we have considered so far. For instance, in SEIRS (or also SIRS) models, the immunity is temporary, \textit{i.e.} people in the class $\rm R$ may come back into the class $\rm S$ at rate $\nu$. Using a similar stochastic extension of this model, it is straightforward that all our results extend, \emph{mutatis mutandis}, to this case as well, albeit with one important difference: the control problem faced by the government now has 5 states variables, namely $(S,E,I,R,Y)$. Even more generally, our approach can readily be adapted to compartmental models considering additional classes: for instance the SIDARTHE (`Susceptible' (S), `Infected' (I), `Diagnosed' (D), `Ailing' (A), `Recognised' (R), `Threatened' (T), `Healed' (H) and `Extinct' (E)) model investigated in \citeauthor*{giordano2020modelling} \cite{giordano2020modelling} for COVID-19. Of course the price to pay is that the number of state variables in the government's problem will increase with the number of compartments, and numerical procedures to solve the HJB equation will become more delicate to implement, and could be based on neural networks.

\end{document}